\documentclass[journal]{IEEEtran}
\usepackage{amsmath,amssymb}
\usepackage[dvips]{graphicx}
\usepackage{amsfonts}
\usepackage[mathscr]{eucal}
\usepackage{latexsym}
\usepackage{amsthm}
\usepackage{exscale}
\usepackage[mathscr]{eucal}
\usepackage{bm}
\usepackage[dvipsnames]{color}
\usepackage{cases}
\usepackage{epsfig}
\usepackage[center,small]{caption}
\usepackage{algorithm}
\usepackage{algorithmic}
\usepackage[verbose,nospace,sort]{cite}
\usepackage{tabularx}
\usepackage{multirow}
\usepackage{balance}
\usepackage{url}
\usepackage{graphicx}
\usepackage{subcaption}
\usepackage[export]{adjustbox}

\scrollmode
\newtheoremstyle{uprightplain}
  {}
  {}
  {\normalfont}
  {}
  {\bfseries}
  {.}
  { }
  {}

\theoremstyle{uprightplain}
\newtheorem{example}{Example}

\theoremstyle{plain}

\newtheorem{theorem}{Theorem}

\newtheorem{lemma}{Lemma}
\newtheorem{corollary}{Corollary}

\IEEEoverridecommandlockouts

\begin{document}
\title{Hybrid Mono- and Bi-static OFDM-ISAC via BS-UE Cooperation: Closed-Form CRLB and Coverage Analysis}
\author{Xiaoli~Xu, \emph{Member, IEEE,}  and Yong Zeng, \emph{Fellow, IEEE}
\thanks{X. Xu and Y. Zeng are with the National Mobile Communications Research Laboratory, Southeast University, Nanjing 210096, China. Y. Zeng is also with the Pervasive Communication Research Center, Purple Mountain Laboratories, Nanjing 211111, China (email: {xiaolixu, yong\_zeng}@seu.edu.cn). {\it (Corresponding author: Yong Zeng.)}}
}

\maketitle
\begin{abstract}
This paper proposes a hybrid mono- and bi-static sensing framework, by leveraging the base station (BS) and user equipment (UE) cooperation in integrated sensing and communication (ISAC) systems. This scheme is built on 3GPP-supported sensing modes, and it does not incur any extra spectrum cost or inter-cell coordination. To reveal the fundamental  performance limit of the proposed hybrid sensing mode,  we derive closed-form Cram\'{e}r-Rao lower bound (CRLB) for sensing target localization and velocity estimation, as functions of target and UE positions. The results reveal that  significant performance gains can be achieved over the purely mono- or bi-static sensing, especially when the BS-target-UE form a favorable geometry, which is close to a right triangle. The analytical results are validated by simulations using effective parameter estimation algorithm and weighted mean square error (MSE)  fusion method. Based on the derived sensing bound, we further analyze  the sensing coverage by varying the UE positions, which shows that sensing coverage first improves then degrades as the BS-UE separation increases. Furthermore, the sensing accuracy for a potential target with best UE selection is derived as a function of the UE density in the network.
\end{abstract}

\section{Introduction}
The International Telecommunication Union Radiocommunication Sector (ITU-R) has identified integrated sensing and communication (ISAC) as one of the six key application scenarios for 6G networks \cite{ITU}. At the recent 3GPP Physical Layer Working Group meeting, Orthogonal Frequency-Division Multiplexing (OFDM) was established as the baseline waveform for 6G, with broad consensus among participants to continue using OFDM and its variants as the dominant waveform for both communication and sensing in 6G systems \cite{3gpp_ran1_122bis}. From a communication perspective, OFDM offers excellent spectral efficiency and robustness against inter-symbol interference (ISI). From a sensing standpoint, it also provides several compelling advantages, including flexible time-frequency resource allocation, efficient decoupled estimation of delay and Doppler shift \cite{4977002}, low ranging sidelobes \cite{liu2024ofdm}, and a near ``thumbtack-shaped" ambiguity function \cite{cao2016feasibility}. Owing to these dual-functional benefits, OFDM-based ISAC systems have attracted significant research interest from both academia and industry \cite{DaiQiang}.

The idea of using OFDM for radar sensing was first introduced in \cite{OFDMsensing}, and its integration with communication was subsequently explored in \cite{4058251}. The sensing performance of OFDM-ISAC is influenced by several key parameters. Specifically, the unambiguous sensing range and maximum detectable Doppler shift are determined by the OFDM subcarrier spacing and symbol duration, respectively \cite{10706597}. The random communication symbols carried by the ISAC signal act as a random mask that greatly attenuate the ISI in the sensing profile, which enhances the sensing range \cite{ISIpaper}. In terms of sensing resolution, the range resolution improves with increased bandwidth, while Doppler resolution enhances as the coherent processing interval grows longer. When the periodogram algorithm is used to estimate target parameters, the resolution is limited by the Rayleigh resolution limit. In contrast, super-resolution algorithms, such as MUltiple SIgnal Classification (MUSIC), can surpass this limit, with their ultimate resolution constrained by statistical resolution limit \cite{10474102}, which is related to the Cram\'{e}r-Rao Lower Bound (CRLB) for the parameter of interest \cite{1420803}. Besides the resolution limit, the target localization accuracy is also related with the CRLB of parameter estimation \cite{9652071}.

CRLB is a widely used metric for evaluating sensing performance. For instance, in \cite{fang2025}, it serves as the performance criterion for optimizing beamforming in mono-static sensing configurations. Similarly, \cite{9814645} employs the CRLB to assess sensing accuracy and to identify the optimal bi-static transmitter-receiver pair. In \cite{10419729}, the CRLB is derived for target localization in a cloud-radio access network (C-RAN), where a single transmitter and multiple sensing receivers are deployed. The CRLB for near-field sensing with extremely-large antenna is derived in \cite{b8}.  Recently, \cite{pucci2025} presents explicit CRLB expressions  for estimating key parameters (i.e., angle of arrival (AoA), propagation delay and Doppler shift)  of a single target in OFDM-ISAC systems, covering  both mono-static and bi-static sensing modes. Building upon the CRLB of these signal-level parameters, the authors further derive the target position error bound (PEB) and velocity error bound (VEB) using the expected Fisher information matrix and a Jacobian transformation.


In the context of cellular network based ISAC systems,  mono-static and bi-static setups represent two fundamental sensing paradigms.  Specifically, 3GPP has defined six canonical ISAC deployment scenarios \cite{DaiQiang}\cite{3GPP_ISAC}: (1) BS mono-static, (2) BS-BS bi-static, (3) UE-BS bi-static, (4) UE mono-static, (5) UE-UE bi-static, and (6) BS-UE bi-static. Mono-static and bi-static sensing share similarities with conventional radar systems in both architecture and signal processing, offering advantages such as simplified coordination and low computational complexity. However, given the limited spectrum resources typically allocated for sensing and their limited coverage range, purely mono-static or bi-static configurations may be insufficient to meet the demanding requirements of practical applications, such as autonomous driving and drone surveillance, in terms of sensing coverage and accuracy.

To better leverage the ubiquitously deployed BSs and the abundant UEs in wireless systems, some preliminary research efforts have been devoted to cooperative multi-static sensing \cite{pucci2025,10226276,10636720,10791445}. In particular, \cite{pucci2025} proposes two general cooperative sensing frameworks: (i) multiple BSs transmit orthogonal sensing signals toward a common target and receive their respective signals for target sensing, and (ii) a single BS transmits the sensing signal while multiple BSs jointly receive the echoes and collaboratively perform parameter estimation. The aforementioned works have been focused on the cooperation of BSs. This is because BSs usually possess higher sensing and computational capabilities than UEs, owing to their larger antenna arrays, backhaul connection and prior knowledge of the transmitted signal.  However, BS deployment is usually static and the backscattered signal received at the BS is usually weak for distant target, as it traverses a dual-hop path. In contrast, cellular networks serve a large number of UEs with dynamic and flexible locations. Hence, it is likely that at least some UEs will receive a stronger sensing signal when the target is in close proximity to them.

Therefore, in this paper, we propose a hybrid mono- and bi-static sensing mode based on the cooperation between BS and UE. Specifically, as illustrated in Fig.~\ref{F:model} while the BS communicates with the UE, it simultaneously performs target sensing using the backscattered echo, corresponding to the BS mono-static sensing mode defined in \cite{3GPP_ISAC}.  Meanwhile, the UE receives the downlink communication signal. After successfully decoding the data, the UE removes the random communication symbols from its received signal and carries out target parameter estimation, corresponding to the BS-UE bi-static sensing mode. Subsequently, the UE feeds its estimation results back to the BS to enhance overall sensing accuracy. Compared with existing cooperative sensing approaches, the proposed BS-UE cooperative framework offers the following advantages:
\begin{itemize}
\item It incurs no additional spectrum overhead compared to a pure communication system, as sensing is performed using the communication signal.
\item It is built upon the 3GPP-supported sensing modes and does not require inter-cell coordination, thereby ensuring low implementation complexity.
\end{itemize}

Based on CRLB for signal-level parameter estimation derived in \cite{pucci2025}, we derive an explicit CRLB expression for target localization in the proposed hybrid mono- and bi-static sensing mode. Our analysis reveals that localization accuracy can be significantly improved by fusing the delay (i.e., time-of-arrival) estimation from the UE. This enhancement is particularly pronounced when the BS, target, and UE form a favourable geometry, which is close to a right triangle. Besides, whereas pure mono- or bi-static sensing can only estimate the radial velocity, the hybrid sensing mode enables estimation of both radial and tangential velocities. We derive the explicit CRLB for velocity estimation, and hence the velocity estimation bound (VEB), as functions of the target and UE positions. The analysis reveals that the estimation accuracy is independent of the ground-truth  velocity value, and  the optimal bi-static angle is derived as a function of the ratio between Doppler shift estimation quality at the BS and UE.  We verify the analytical results through simulations using a practical sequential parameter estimation method based on the fast Fourier transform (FFT) and a weighted mean square error (MSE) fusion method. A close match between the analytical PEB/VEB and the simulated position/velocity estimation errors is observed when the signal-to-noise ratio (SNR) is sufficiently high. The analysis provides useful guidance for practical network design, and two use cases are discussed, i.e., network sensing coverage analysis and the impact of user density on sensing performance. The main contributions of this paper are summarized as follows.
\begin{itemize}
\item {\it Closed-Form CRLB Expression}: We derive the CRLB for target position and velocity estimation under hybrid mono- and bi-static sensing. Closed-form expressions are provided as functions of the target and UE positions. Our analysis reveals that the improvement in sensing accuracy achieved by fusing UE-based estimates depends critically on both the quality of the UE's measurements and the geometric configuration of the BS-Target-UE triangle. We further derive the fundamental sensing limit when the UE approaches the target arbitrarily closely, and validate our analytical results through simulations.

\item {\it Coverage Analysis}: Building on the CRLB analysis, we compare the sensing coverage of conventional BS mono-static sensing and hybrid sensing for a fixed UE position. Sensing coverage is defined as the area where the targets may locate so that the PEB falls below a given threshold. Our results show that sensing coverage initially expands as the distance between the BS and UE increases, but begins to shrink once this separation exceeds an optimal value. This non-trivial behavior offers valuable guidance for network design, particularly in the strategic placement of dedicated sensing nodes within existing cellular infrastructure.

\item {\it Impact of User Density}: We further assess sensing performance for a given target by selecting the best cooperating UE from a spatially distributed population. Interestingly, the results demonstrate that the target's PEB can satisfy (i.e., fall below) a predefined requirement if at least one UE resides within an "8"-shaped admissible region surrounding the target. We analytically derive the area of this region as a function of the target's required PEB, thereby linking user density to achievable sensing performance.
\end{itemize}

The rest of this paper is organized as follows. The system model of the proposed hybrid mono- and bi-static sensing is presented in Section~\ref{sec:model}.  Section~\ref{sec:analysis} derives the CRLB for target position and velocity estimation in the hybrid sensing mode. The analysis is verified with practical estimation algorithm in Section~\ref{sec:verify}. In Section~\ref{sec:utilization}, the analytical results are used to derive the sensing coverage for given cooperating UE position and inform the cooperating sensing UE selection for given target. Finally, this paper is concluded in Section~\ref{sec:conclusion}.

\emph{Notation}: Boldfaced lower- and upper- case characters denote  vectors and matrices, respectively. For matrix $\mathbf{A}$, we use $\mathbf{A}^T$, $\mathbf{A}^H$, $\mathrm{Tr}(\mathbf{A})$ and $[\mathbf{A}]_{i,j}$ to denote the transpose, Hermitian transpose, trace and the $(i,j)$th element of matrix $\mathbf{A}$. $\mathbf{0}_{i\times j}$ is an all-zero matrix of size $i\times j$. $\succeq$ is the positive semi-definite inequality of matrix. $\mathbb{C}^{M\times N}$ and $\mathbb{R}^{M\times N}$ signifies the spaces of $M\times N$ complex and real matrices. $\mathrm{Rect}(\frac{t}{T})$ is the rectangular window function centered at $t=0$, with window length $T$. $\mathbb{E}\left[\cdot\right]$ and $\mathrm{var}(\cdot)$ denote the expectation and variance of a random variable.

\section{System Model}\label{sec:model}
As shown in Fig.~\ref{F:model}, we consider a hybrid mono- and bi-static ISAC system enabled by BS-UE cooperation, where the BS transmits the ISAC signal and senses targets by receiving the echoes. This corresponds to the BS mono-static mode. Meanwhile, the UE decodes the communication data from the received ISAC signal, then removes the decoded data to extract sensing information about the targets, operating in the BS-UE bi-static mode. In this hybrid ISAC framework, the UE feeds back its local sensing results to the BS, enabling enhanced sensing accuracy through cooperative fusion.

\begin{figure}[htb]
\centering
\includegraphics[width=0.4\textwidth]{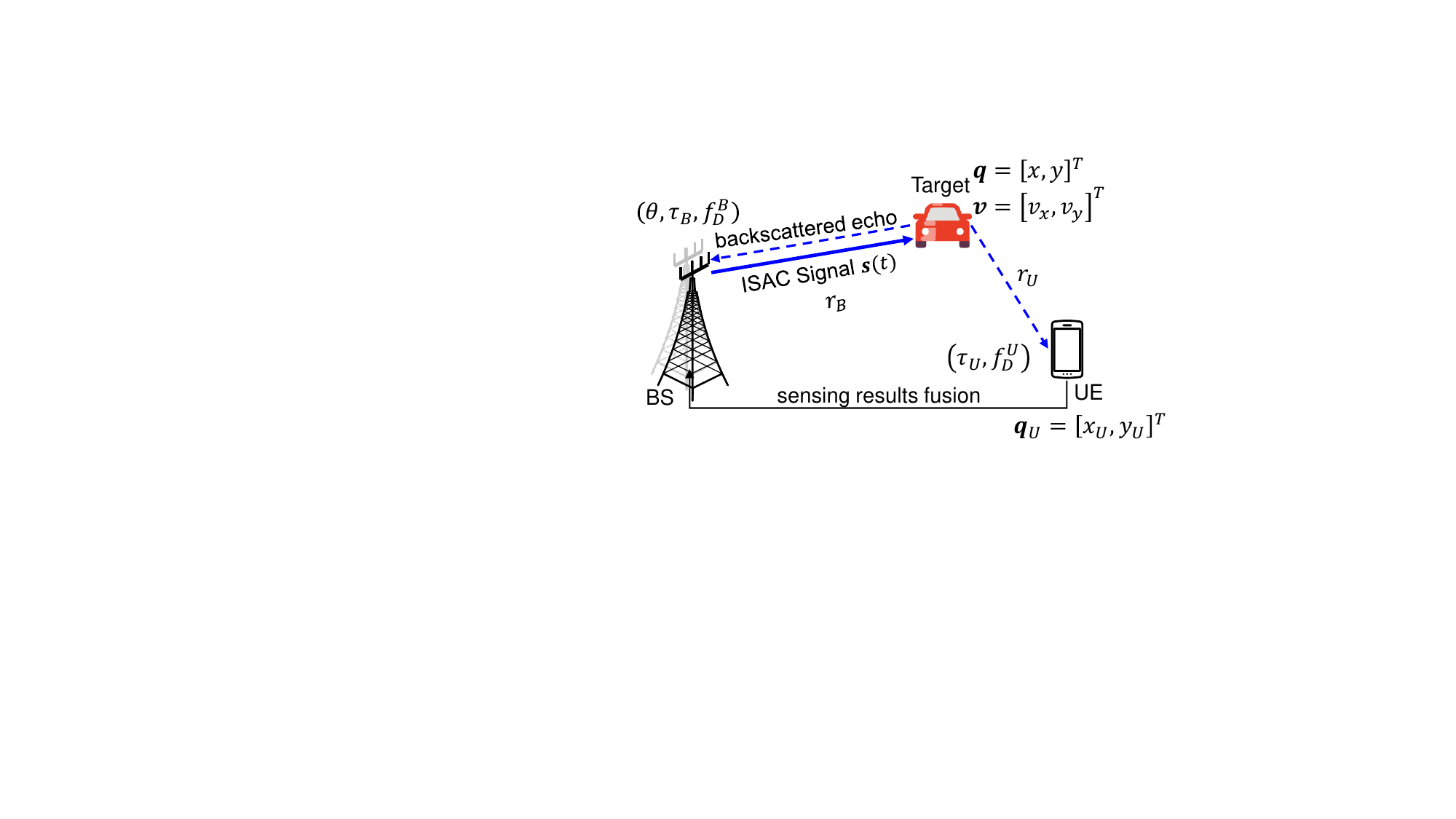}
\caption{The hybrid mono- and bi-static ISAC mode with BS-UE cooperation.}
\label{F:model}
\end{figure}
Without loss of generality, we assume that the BS is located at the origin and operates in full-duplex mode. The transmit antenna consists of an $N_T$-element uniform linear array (ULA), and the receive antenna is an $N_R$-element ULA. Both ULAs are aligned along the $y$-axis, with their array centers coinciding with the origin.  The system operates at carrier frequency $f_c$, and the inter-element spacing is set to half the wavelength,  i.e., $d=\lambda_c/2$, where $\lambda_c=c/f_c$ and $c$ is the speed of light.  In the BS-centered coordinate system, the UE is located at $\mathbf{q}_U=[x_U, y_U]^T$, and is equipped with a single antenna. Denote the target position and velocity by $\mathbf{q}=[x,y]^T$ and $\mathbf{v}=[v_x,v_y]^T$, respectively, which need to be estimated based on the backscattered echo received by the BS and the scattered signal received by the UE.

\subsection{Signal Model}
The ISAC signal sent by the BS is modulated using OFDM, with $K$ subcarriers and $M$ OFDM symbols \cite{DaiQiang}. Denote by $d_{km}$  the communication symbol for the UE carried by the $k$th subcarrier of the $m$th symbol and assume $\mathbb{E}\{|d_{km}|^2\}=1$, where $k=0,...,K-1$ and $m=0,...,M-1$. We consider the basic case where the  symbol is transmitted via all antennas with beamforming vector $\mathbf{w}\in\mathbb{C}^{N_T\times 1}$, subject to $\|\mathbf w\|=1$. With the total transmit power denote by $P_T$, the power per subcarrier is $P_{\mathrm{avg}}=\frac{P_T}{K}$. The signal samples at the $N_T$ transmit antennas can be written as $\mathbf{x}_{km}=\sqrt{P_{\mathrm{avg}}}\mathbf{w}d_{km}$. Hence, the baseband transmit signal can be written as
\begin{align}
\mathbf{s}(t)=\sum_{m=0}^{M-1}\sum_{k=0}^{K-1}\mathbf{x}_{km}e^{j2\pi k\Delta f(t-mT_s)}\mathrm{rect}\left(\frac{t-mT_s}{T_s}\right),
\end{align}
where $\Delta f$ is the OFDM subcarrier spacing, $T=1/\Delta f$ is the OFDM symbol duration, and $T_s=T+T_{cp}$ is the OFDM symbol duration including the cyclic prefix (CP).

Assume that the interference from background clutter has been successfully removed. Then, the backscattered channel observed by the BS only contains the reflection from the target, with the complex channel gain, delay and Doppler denoted by $\alpha_B$, $\tau_B$ and $f_{D}^B$, respectively. The received signal  after OFDM demodulation can be expressed as \cite{DaiQiang}
\begin{align}
\mathbf{y}_{km}^B=\mathbf{H}_{km}^B\mathbf{x}_{km}+\mathbf{z}_{km}=\sqrt{P_{\mathrm{avg}}}\mathbf{H}_{km}^B\mathbf{w}d_{km}+\mathbf{z}_{km}, \label{eq:RxBS}
\end{align}
where $\mathbf{H}_{km}^B\in\mathbb{C}^{N_R\times N_T}$ is the effective channel matrix given by
\begin{align}
\mathbf{H}_{km}^B=\alpha_B\mathbf{b}(\theta)\mathbf{a}^{H}(\theta)e^{-j2\pi k\Delta f\tau_B}e^{j2\pi mT_sf_{D}^{B}}, \label{eq:HB}
\end{align}
with $\mathbf{a}(\theta)\in\mathbb{C}^{N_T\times 1}$ and $\mathbf{b}(\theta)\in\mathbb{C}^{N_R\times 1}$ being the transmit and receive steering vectors of the ULA, and $\theta$ is the angle of the target with respect to the antenna array of the BS, as shown in Fig.~\ref{F:model}.  $\mathbf{z}_{km}\sim\mathcal{CN}(\mathbf{0},\sigma_N^2\mathbf{I}_{N_{R}})$ is the additive white Gaussian noise (AWGN).

Similarly, the received signal at the UE via the target scattering path can be written as
\begin{align}
y_{km}^U&=\mathbf{h}_{km}^H\mathbf{x}_{km}+v_{km}=\sqrt{P_{\mathrm{avg}}}\mathbf{h}_{km}^H\mathbf{w}d_{km}+z_{km},
\end{align}
where $\mathbf{h}_{km}=\alpha_U\mathbf{a}(\theta)e^{-j2\pi k\Delta f\tau_U}e^{j2\pi f_{D}^UmT_s}$ is the  channel observed by UE and $z_{km}\sim\mathcal{CN}(0,\sigma_N^2)$ is the AWGN noise.

\subsection{SNR Analysis}
Based on \eqref{eq:RxBS} and \eqref{eq:HB}, the SNR of the echo signal received by each antenna element of the BS, at a single resource element can be calculated as
\begin{align}
\tilde\Upsilon_B=P_{\mathrm{avg}}\frac{|\alpha_B|^2|\mathbf{a}^H(\theta)\mathbf{w}|^2}{\sigma_N^2}. \label{eq:SNRmAnt}
\end{align}
Considering the dual-hop path loss, we have
\begin{align}
|\alpha_B|^2=\frac{c^2\kappa_B}{(4\pi)^3f_c^2r_B^4},
\end{align}
where $\kappa_B$ is the radar cross-section (RCS) observed by the BS, $r_B$ is the distance between the target and BS, and $f_c$ is the carrier frequency. For notational convenience, denote the beamforming gain at the direction of the target by $\beta\triangleq |\mathbf{a}^H(\theta)\mathbf{w}|^2$, where $\beta\leq N_T$ and the equality holds if the beam happens to be in the direction of the target, i.e., $\mathbf{w}=\frac{1}{\sqrt{N}_T}\mathbf{a}(\theta)$.  Hence, the SNR in \eqref{eq:SNRmAnt} can be written as
\begin{align}
\tilde{\Upsilon}_B=\frac{\beta P_{\mathrm{avg}}c^2\kappa_B}{(4\pi)^3f_c^2r_B^4\sigma_N^2}. \label{eq:SNRm}
\end{align}

Similarly, we can obtain the SNR at the UE as
\begin{align}
\tilde\Upsilon_U=\frac{\beta P_{\mathrm{avg}}c^2\kappa_U}{(4\pi)^3f_c^2r_B^2r_U^2\sigma_N^2}, \label{eq:SNRu}
\end{align}
where $\kappa_U$ is the RCS observed by UE, and $r_U$ is the distance between the target and the UE.

For the purpose of sensing, the random data symbols $\{d_{km}\}$ must be removed from the received signal $\mathbf{y}_{km}^B$ and $y_{km}^U$, in order to extract the pure channel response induced by the scatterers. This is typically accomplished through either matched filtering or data division. When a constant-envelope modulation scheme is used to generate $d_{km}$, e.g., phase shift keying (PSK), no SNR penalty is incurred. Otherwise, an SNR penalty arises after data division,  denoted by $\eta$, where $\eta\geq 1$\cite{LiuFan2022}. Therefore, the SNR of the data-free sensing signal at the BS and UE are given by
\begin{align}
\Upsilon_B=\tilde\Upsilon_B/\eta, \textnormal{ and } \Upsilon_U=\tilde\Upsilon_U/\eta. \label{eq:SNR}
\end{align}

\subsection{Target Parameter Estimation}
The log likelihood function for receiving $\mathbf{y}_{km}^B$ given the target-related parameter set $\Theta_B=\{\alpha_B, \tau_B,\theta,f_{D}^B\}$ can be written as \eqref{eq:likelihoodBS}, shown on top of the next page. Similarly, the log likelihood function for receiving $y_{km}^U$  at the UE is a function of $\Theta_U=\{\alpha_U,\tau_U,f_{D}^U\}$ given in \eqref{eq:likelihoodUE}.
\begin{figure*}
\begin{align}
p_1(\mathbf{y}_{km}^B;\Theta_B)&=\prod_{k=0}^{K-1}\prod_{m=0}^{M-1}\frac{1}{(\pi\sigma_N^2)^{N_R}}\exp\left(-\frac{1}{\sigma_N^2}\left\|\mathbf{y}_{km}^B-\mathbf{b}(\theta)d_{km}\alpha_B \beta e^{-j2\pi k\Delta f\tau_B}e^{j2\pi mT_sf_{D}^{B}}\right\|^2\right),\label{eq:likelihoodBS}\\
p_2({y}_{km}^U;\Theta_U)&=\prod_{k=0}^{K-1}\prod_{m=0}^{M-1}\frac{1}{\pi\sigma_N^2}\exp\left(-\frac{1}{\sigma_N^2}\left|y_{km}^U-d_{km}\alpha_U\beta e^{-j2\pi k\Delta f\tau_U}e^{j2\pi mT_sf_{D}^{U}}\right|^2\right),\label{eq:likelihoodUE}
\end{align}
\end{figure*}
The Fisher information matrix for estimation at the BS and UE, denoted by $\mathbf{I}(\Theta_B)\in\mathbb{R}^{3\times 3}$ and $\mathbf{I}(\Theta_U)\in\mathbb{R}^{2\times2}$,  can be calculated from \eqref{eq:likelihoodBS} and \eqref{eq:likelihoodUE} as
\begin{align}
\left[\mathbf{I}(\Theta_B)\right]_{ij}=-\mathbb{E}\left[\frac{\partial^2\ln p_1(\mathbf{y}_{km}^B;\Theta_B)}{\partial \Theta_B(i)\partial \Theta_B(j)}\right], \label{eq:FisherB}\\
\left[\mathbf{I}(\Theta_U)\right]_{ij}=-\mathbb{E}\left[\frac{\partial^2\ln p_2(y_{km}^U;\Theta_U)}{\partial \Theta_U(i)\partial \Theta_U(j)}\right], \label{eq:FisherU}
\end{align}
where $\Theta_B(i)$ is the $i$th element in the set $\Theta_B$, with $i\in\{1,2,3,4\}$, and $\Theta_U(i)$ is the $i$th element in the set $\Theta_U$, with $i\in\{1,2,3\}$. The details for deriving each element of $\mathbf{I}(\Theta_B)$ and $\mathbf{I}(\Theta_U)$ can be found from the appendix of \cite{pucci2025}, which are omitted here for brevity.

Since $\Theta_B$ and $\Theta_U$ are independent parameters estimated from the independent observations, the total Fisher information matrix for the whole parameter set $\Theta \triangleq \Theta_B\bigcup \Theta_U$ is denoted by $\mathbf{I}(\Theta)\in\mathbb{R}^{7\times 7}$, written as
\begin{align}
\mathbf{I}(\Theta)=\left[
\begin{matrix} \mathbf{I}(\Theta_B) & \mathbf{0}_{4\times 3} \\ \mathbf{0}_{3\times 4} & \mathbf{I}(\Theta_U)
\end{matrix}\right]. \label{eq:FisherAll}
\end{align}

The ultimate goal is to estimate the target position $\mathbf{q}$ and velocity $\mathbf{v}$, denoted as  $ \Phi\triangleq \{\mathbf q, \mathbf{v}\}$. Note that the parameter set $\Theta$ can be written as multi-dimensional functions as $ \Phi$, i.e., $\Theta=\mathbf{g}( \Phi)$. According to the rule of parameter transformation \cite{kay1993fundamentals}, the Fisher information matrix for $ \Phi$ can be written as
\begin{align}
\mathbf{I}(\Phi)=\mathbf{J}^T\mathbf{I}(\Theta)\mathbf{J},
\end{align}
where $\mathbf{J}$ is the Jacobian matrix of transformation, with $[\mathbf{J}]_{ij}=\frac{\partial \Theta(i)}{\partial{\Phi}(j)}$. The estimation error matrix of ${\Phi}$ with any unbiased estimator is bounded as
\begin{align}
\mathrm{Cov}({\hat{\Phi}})\triangleq\mathbb{E}\left[(\hat{\Phi}-\Phi)(\hat{\Phi}-\Phi)^T\right]\succeq \left(\mathbf{J}^T\mathbf{I}(\Theta)\mathbf{J}\right)^{-1}. \label{eq:jointEst}
\end{align}

 In OFDM-ISAC, joint estimation of target position $\mathbf{q}$  and velocity $\mathbf{v}$, with performance bounded by \eqref{eq:jointEst},  can be achieved through the joint estimation of delay and Doppler \cite{ISACMLEsti}. However, this approach entails significant computational complexity \cite{DaiQiang}. In this paper, we consider the decoupled estimation of $\mathbf{q}$ and $\mathbf{v}$ for simplicity.


\section{CRLB for Hybrid Mono- and Bi-static Sensing}\label{sec:analysis}
The section derives the CRLB for estimating the target position and velocity by considering the related parameters in $\Theta$, while treating others as unknown variables. In particular, target position $\mathbf{q}$ is estimated from the obtained delay and AoA, i.e., $\{\tau_B, \theta, \tau_U\}$, and the target velocity $\mathbf{v}$ is estimated from the obtained Doppler $\{f_D^B, f_D^U\}$.


\subsection{CRLB for Target Localization}

The target position $\mathbf q$ can be estimated from $\Theta_1=\{\tau_B, \theta, \tau_U\}$ based on the transformation function $\mathbf{g}_1(\mathbf q)$, where
\begin{align}
\left[\begin{matrix}\tau_B \\ \theta \\ \tau_U\end{matrix}\right]=\mathbf{g}_1(\mathbf{q})=
\left[\begin{matrix}\frac{2}{c}\|\mathbf{q}\| \\ \arctan\left(\frac{\mathbf{q}(2)}{\mathbf{q}(1)}\right) \\ \frac{1}{c}(\|\mathbf{q}\|+\|\mathbf{q}-\mathbf{q}_U\|)\end{matrix}\right].\label{eq:g1}
\end{align}

Since $\Theta_1$ is a subset of $\Theta$, all the necessary information to derive the estimation variance bound for the parameter $\Theta_1$ is retained by the expected Fisher Information Matrix (EFIM) $\mathbf{I}(\Theta_1)$, which can be calculated from $\mathbf{I}(\Theta)$ according to the definition of EFIM in \cite{5571889}. Another way to construct the EFIM is to divide $\Theta_1$ into two parts, i.e, those from BS $\Theta_{1,B}=\{\tau_B, \theta\}$ and the other parameter from UE, i.e., $\tau_U$. The EFIM for $\Theta_{1,B}$ has been derived in \cite{pucci2025}, which is given by
\begin{align}
\mathbf{I}(\Theta_{1,B})=\left[\begin{matrix}\mathcal{I}_{\tau_B} & \mathcal{I}_{\tau_B,\theta} \\
\mathcal{I}_{\theta,\tau_B} & \mathcal{I}_{\theta}
\end{matrix}\right]\label{eq:Itheta1B}
\end{align}
where
\begin{align}
&\mathcal{I}_{\tau_B}=\frac{2\pi^2\Delta f^2MK(K^2-1)N_R \Upsilon_B}{3},\label{eq:FishTauB}\\
&\mathcal{I}_{\theta}=\frac{\pi^2 KM(N_R^2-1)N_R\Upsilon_B\cos^2(\theta)}{6},\label{eq:FishTheta0}\\
&\mathcal{I}_{\tau_B,\theta}=\mathcal{I}_{\theta,\tau_B}=0.
\end{align}

The Fisher information for $\tau_U$ is well-established for the bi-static sensing \cite{pucci2025}, i.e.,
\begin{align}
&\mathcal{I}_{\tau_U}=\frac{2\pi^2\Delta f^2MK(K^2-1) \Upsilon_U}{3},\label{eq:FishTauU0}
\end{align}

Consider the independence between the BS and UE estimation, the EMIF for the parameter set $\Theta_1$ is
\begin{align}
&\mathbf{I}(\Theta_1)
=\left[\begin{matrix}\mathcal{I}_{\tau_B} & 0 & 0\\
0 & \mathcal{I}_{\theta}  & 0 \\
0 & 0 & \mathcal{I}_{\tau_U}
\end{matrix}\right]\approx
\frac{2\pi^2MK}{3}\nonumber\\
&\cdot\left[\begin{matrix}W^2N_R\Upsilon_B & 0 & 0\\
0 & \frac{(N_R^2-1)\cos^2\theta}{4}N_R\Upsilon_B  & 0 \\
0 & 0 & W^2\Upsilon_U
\end{matrix}\right],
\label{eq:Itheta1}
\end{align}
where $W=\Delta fK$ is the system bandwidth. The approximation is made by considering a sufficiently large number of subcarriers, i.e., $K\gg 1$.



Next, we consider the Jacobian matrix of transformation, calculated from \eqref{eq:g1}. For convenience, we use the $xy$-coordinates to represent the target and UE position, i.e., $\mathbf{q}=[x,y]^T$ and $\mathbf{q}_U=[x_U,y_U]^T$, which renders
\begin{align}
\mathbf{J}_{\mathbf{q}}&=\left[\begin{matrix}\frac{\partial\tau_B}{\partial x} & \frac{\partial\tau_B}{\partial y} \\
\frac{\partial\theta}{\partial x} & \frac{\partial\theta}{\partial y} \\ \frac{\partial\tau_U}{\partial x} & \frac{\partial\tau_U}{\partial y}\end{matrix}\right]
=\left[\begin{matrix}\frac{2x}{cr_B} & \frac{2y}{cr_B} \\
\frac{-y}{x^2+y^2} & \frac{x}{x^2+y^2} \\ \frac{x}{cr_B}+\frac{x-x_U}{cr_U} & \frac{y}{cr_B}+\frac{y-y_U}{cr_U}\end{matrix}\right].\label{eq:Jq}
\end{align}

Note that the coordinate of UE appears in the Jacobian matrix and hence the CRLB of target localization is a function of both the target position $\mathbf{q}$ and UE position $\mathbf{q}_U$. We denote the CRLB for target localization with hybrid sensing by $C_{h}(\mathbf{q},\mathbf{q}_U)$, and we have
\begin{align}
\mathrm{var}(\|\mathbf{q}-\mathbf{\hat q}\|^2)\geq
C_{h}(\mathbf{q},\mathbf{q}_U)=\mathrm{Tr}\left((\mathbf{J}_{\mathbf{q}}^T\mathbf{I}(\Theta_1)\mathbf{J}_{\mathbf{q}})^{-1}\right)
\end{align}
where $\mathbf{\hat q}$ is the unbiased estimation of target position $\mathbf{q}$ with the hybrid mono- and bi-static sensing. 

\begin{theorem}
The CRLB for target localization in OFDM-ISAC with BS-UE cooperation can be expressed in closed-form in terms of the target location $\mathbf q$ and UE location $\mathbf q_U$, given by
\begin{align}
&C_{h}(\mathbf{q},\mathbf{q}_U)
=\nonumber\\
&
\frac{4c^2r_B^2\mathcal{I}_{\tau_B}+c^4\mathcal{I}_{\theta}+2c^2r_B^2\left(1+\cos\psi\right)\mathcal{I}_{\tau_U}}
{4c^2\mathcal{I}_{\tau_B}\mathcal{I}_{\theta}+4r_B^2\mathcal{I}_{\tau_B}\mathcal{I}_{\tau_U}\sin^2\psi+c^2\mathcal{I}_{\theta}\mathcal{I}_{\tau_U}\left(1+\cos\psi\right)^2},
\label{eq:CRLBjoint_polar}
\end{align}
where $\mathcal{I}_{\tau_B}$, $\mathcal{I}_{\theta}$, and $\mathcal{I}_{\tau_U}$ are Fisher information for estimating individual parameters, given in \eqref{eq:FishTauB}, \eqref{eq:FishTheta0} and \eqref{eq:FishTauU0}, respectively. $r_B=\|\mathbf{q}\|$ is the distance between the target and BS, and $\psi=\arccos\left(\frac{\mathbf{q}^T(\mathbf{q}_U-\mathbf q)}{\|\mathbf{q}\|\|\mathbf{q}_U\|}\right)$ is the angle formed by BS-target-UE, as shown in Fig.~\ref{F:geometry}.
\end{theorem}
\begin{proof}
Please refer to Appendix~\ref{A:positionCRLB}.
\end{proof}


\begin{figure}[htb]
  \centering
  \includegraphics[width=0.3\textwidth]{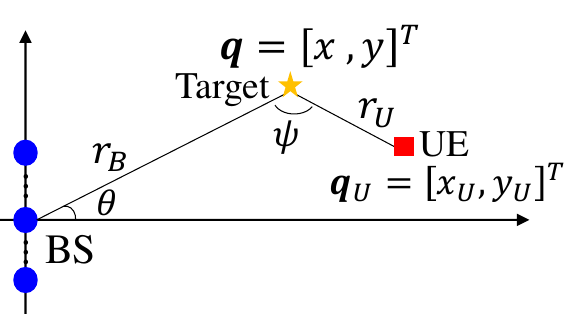}
  \caption{An illustration of BS-target-UE geometry}
  \label{F:geometry}
\end{figure}


Under the special case when the measurements at the UE are not used, the hybrid mono- and bi-static sensing reduces to the BS mono-static sensing. Mathematically, the CRLB of target localization under mono-static sensing can be obtained, say when $\Gamma_U=0$ in \eqref{eq:FisherTauU}, corresponding to $\mathcal{I}_{\tau_U}=0$ in \eqref{eq:CRLBjoint_polar}. Furthermore, considering that $\mathcal{I}_{\tau_B}$ is related with target position by the SNR $\Upsilon_B$ and ${\mathcal{I}_{\theta}}$ is related with target position via $\Upsilon_B$ and AoA $\theta$, we extract the mono-static sensing CRLB as a function of target position in Corollary~\ref{cor:mono}.

\begin{corollary}\label{cor:mono}
The CRLB of target localization with BS mono-static  OFDM-ISAC can be written in closed-form in terms of the target position $\mathbf{q}=[x,y]^T=[r_B\cos\theta,r_B\sin\theta]^T$ as
\begin{align}
C_{\mathrm{mono}}(\mathbf{q})=\frac{r_B^2}{{\mathcal{I}_{\theta}}}+\frac{c^2}{4\mathcal{I}_{\tau_B}}
=c_1\frac{r_B^6}{\cos^2\theta}+c_2r_B^4.\label{eq:CRLBm}
\end{align}
where
\begin{align}
c_1&=\frac{6}{\pi^2 N_RKM(N_R^2-1)\bar{\Upsilon}_B}, \label{eq:C1}\\
c_2&=\frac{3c^2}{8\pi^2 N_RKM\Delta f^2(K^2-1)\bar\Upsilon_B}, \label{eq:C2}
\end{align}
are constants independent of target positions, with $\bar{\Upsilon}_B=\frac{\beta P_{\mathrm{avg}}c^2\kappa_B}{(4\pi)^3f_c^2\sigma_N^2\eta}$ being the coefficient of the SNR by excluding the distance factor.
\end{corollary}

Under the special case when the BS has only a single receive antenna, i.e., $N_R=1$, the AoA information cannot be estimated, and hence we have $\mathcal{I}_{\theta}=0$ in \eqref{eq:FishTheta0}. BS mono-static sensing cannot locate the target, as $C_{\mathrm{mono}}(\mathbf{q})\rightarrow \infty$ in \eqref{eq:CRLBm} as $\mathcal{I}_{\theta}\rightarrow 0$. In contrast, the target can be still localized based on the estimated delay at BS and UE, and the CRLB in given by Corollary~\ref{cor:singleAnt}.
\begin{corollary}\label{cor:singleAnt}
For the special case of single receive antenna at the BS, i.e., $N_R=1$, the CRLB of target localization with BS-UE cooperation based the delay estimated at BS and UE is given by
\begin{align}
C_{h,s}(\mathbf{q},\mathbf{q}_U)&=\frac{c^2}{{\mathcal{I}_{\tau_U}}\sin^2\psi}+\frac{c^2}{2(1-\cos\psi)\mathcal{I}_{\tau_B}}\nonumber\\
&=c_3\frac{r_B^2r_U^2}{\sin^2\psi}+c_4\frac{r_B^4}{1-\cos\psi}.\label{eq:CRLBhs}
\end{align}
where
\begin{align}
c_3=\frac{3}{2\pi^2W^2KM\bar{\Upsilon}_U}, \ c_4=\frac{3}{2\pi^2W^2KM\bar{\Upsilon}_B}, \label{eq:C34}
\end{align}
are constants independent of target and UE positions, with $\bar{\Upsilon}_B$ and $\bar{\Upsilon}_U$ being the coefficients of the SNR at BS and UE, by excluding the distance factor.
\end{corollary}



Since UE has a single antenna,  BS-UE bi-static sensing alone cannot locate the target, which is evident by $C_h(\mathbf{q},\mathbf{q}_U)\rightarrow\infty$ as $\mathcal{I}_{\tau_B}\rightarrow 0$ and $\mathcal{I}_{\theta}\rightarrow 0$ in \eqref{eq:CRLBjoint_polar}. However,  fusing the delay estimation at the UE can reduce the CRLB of target localization with the net CRLB reduction given by
\begin{align}
&\Delta_{C}(\mathbf{q},\mathbf{q}_U)=C_{\mathrm{mono}}(\mathbf{q})-C_h(\mathbf{q},\mathbf{q}_U)=\frac{\mathcal I_{\tau_U}}{4\mathcal{I}_{\tau_B}\mathcal{I}_{\theta}}\nonumber\\
&\cdot\frac{16\mathcal{I}_{\tau_B}^2r_B^4\sin^2\psi+c^4\mathcal{I}_{\theta}^2(1+\cos\psi)^2}{4c^2\mathcal{I}_{\tau_B}\mathcal{I}_{\theta}+4r_B^2\mathcal{I}_{\tau_B}\mathcal{I}_{\tau_U}\sin^2\psi+c^2\mathcal{I}_{\theta}\mathcal{I}_{\tau_U}\left(1+\cos\psi\right)^2}.\label{eq:netExp}
\end{align}
 Since all terms in \eqref{eq:netExp} are positive, we can conclude that $\Delta_{C}(\mathbf{q},\mathbf{q}_U)\geq 0$ regardless where the UE and target locate. The quality holds when $\psi=\pi$. This implies that the sensing results from the UE does not improve the position accuracy if the UE locates on the same line as the BS and Target.

In practice, the localization error is usually measured by the PEB, which is defined as the square root of the CRLB. Hence, we can directly obtain the PEB for mono-static sensing and hybrid sensing as
\begin{align}
&PEB_{\mathrm{mono}}(\mathbf{q})=\sqrt{C_{\mathrm{mono}}(\mathbf{q})}, \label{eq:PEBmono}\\
&PEB_h(\mathbf{q},\mathbf{q}_U)=\sqrt{C_h(\mathbf{q},\mathbf{q}_U)}. \label{eq:PEBjoint}
\end{align}


\begin{example}
To interpret the derived localization CRLB, we present the results visually  for different target and UE locations. Throughout the paper, we adopt an ISAC signal modulated based on OFDM with carrier frequency $f_c=24$GHz, subcarrier spacing $\Delta f=120$kHz, number of subcarriers $K=100$, and number of OFDM symbols  $M=14$. The number of transmit and receive antennas are $N_T=N_R=4$, and the distance-independent transmit SNR is identical for the BS and UE, given by $\bar{\Upsilon}_B=\bar{\Upsilon}_U=98$dB,  unless otherwise specified.

Fig.~\ref{F:PEBwUELoc}(a) shows the PEB for potential targets within the square area bounded by $0\leq x\leq 300\ m$ and $-150\leq y\leq 150\ m$ with BS mono-static sensing. It is observed that PEB less than 1m can only be achieved if the targets are within a small area close to the BS. On the other hand, if there is a UE located at $\mathbf{q}_U=[300,0]^T$, the PEB can be greatly reduced, as shown in Fig.~\ref{F:PEBwUELoc}(b).

\begin{figure}[htb]
\centering
\begin{subfigure}{0.24\textwidth}
\centering
\includegraphics[width=\textwidth]{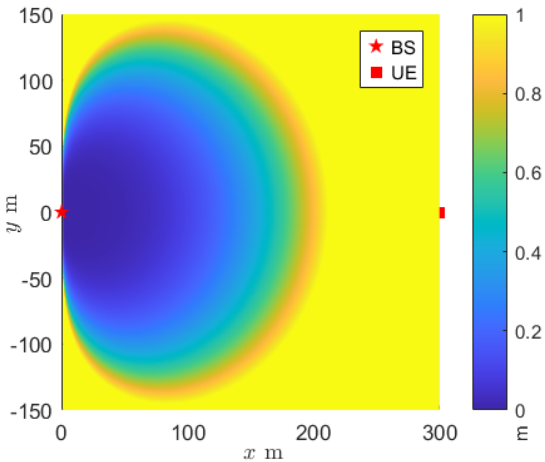}
\caption{BS mono-static sensing}
\end{subfigure}
\begin{subfigure}{0.24\textwidth}
\centering
\includegraphics[width=\textwidth]{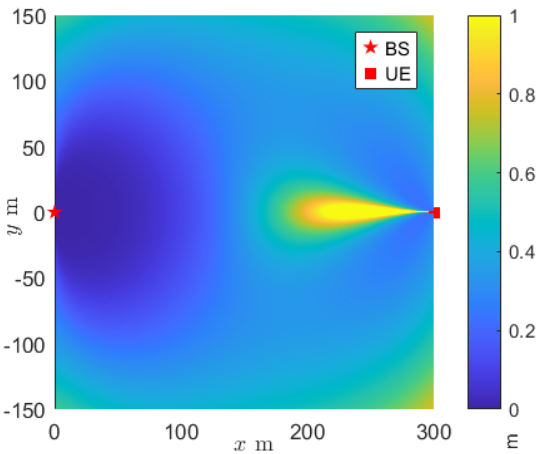}
\caption{hybrid  sensing}
\end{subfigure}
\caption{The comparison of PEB achievable by the BS mono-static versus hybrid sensing with given UE position.}
\label{F:PEBwUELoc}
\end{figure}

Another way to interpret \eqref{eq:CRLBjoint_polar} is to plot the PEB for a fixed target position, e.g., $\mathbf{q}=[200,50]^T$ and varying the cooperating UE position. With BS mono-static sensing, the achievable PEB is 0.96m. Consider a potential UE with position varying in the area $0\leq x_U\leq 300\ m$ and $-100\leq y_U\leq 150\ m$. As shown in Fig.\ref{F:PEBwTargetLoc}(a), the PEB can be reduced to lower than 0.1m if the UE is within the "8" shape region around the target. The "8" shape is enlarged if we scale the PEB limit to 0.6m.

\begin{figure}[htb]
\centering
\begin{subfigure}{0.24\textwidth}
\centering
\includegraphics[width=\textwidth]{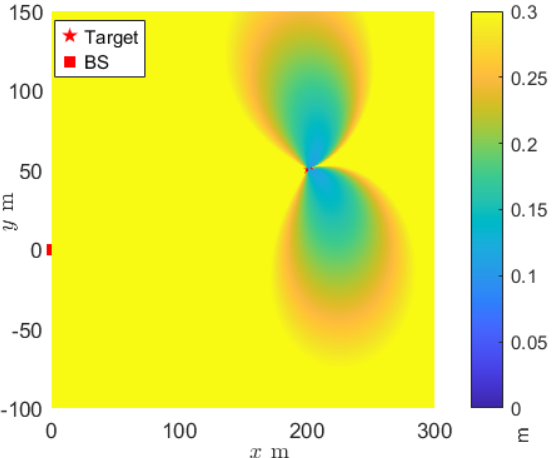}
\caption{$PEB_h$ limit by 0.3m}
\end{subfigure}
\begin{subfigure}{0.24\textwidth}
\centering
\includegraphics[width=\textwidth]{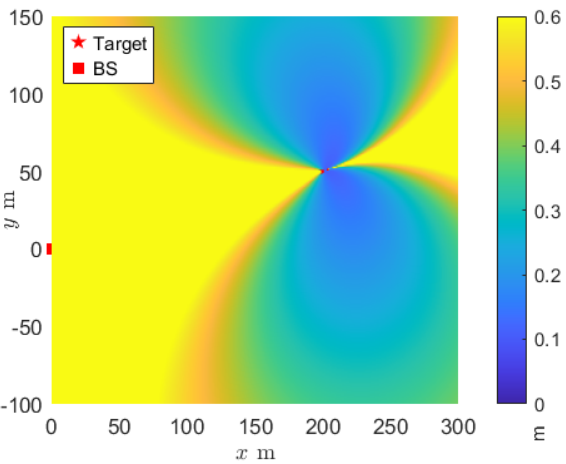}
\caption{$PEB_h$ limit by 0.6m}
\end{subfigure}
\caption{The comparison of PEB with given target position.}
\label{F:PEBwTargetLoc}
\end{figure}
\end{example}

\subsection{Insights and Possible Extension of the Analysis}

Based on the analysis presented in the preceding subsection, we note that the sensing performance of the hybrid mono- and bi-static sensing for OFDM-ISAC is mainly affected by the relative position of target and UE. In particular, we can draw the following insights from the analysis.
\begin{itemize}
\item It is observed that $C_h(\mathbf{q},\mathbf{q}_U)$ is an even function of $\psi$. Consequently, the reduction in CRLB depends only on $|\psi|$. This implies that two UEs symmetric  with respect to the line connecting the BS and the target provide identical contribution to the sensing.
\item If $\mathcal{I}_{\tau_U}\gg \mathcal{I}_{\tau_B}$ and $\mathcal{I}_{\tau_U}\gg \mathcal{I}_{\theta}$,  the CRLB converges to
    \begin{align}
    C_h&(\mathbf{q},\mathbf{q}_U)=\frac{2(1+\cos\psi)}{\frac{4\mathcal{I}_{\tau_B}}{c^2}\sin^2\psi+\frac{\mathcal{I}_{\theta}}{r_B^2}(1+\cos\psi)^2}\nonumber\\
    &=\frac{2}{\left(\frac{4\mathcal{I}_{\tau_B}}{c^2}+\frac{\mathcal{I}_{\theta}}{r_B^2}\right)-\left(\frac{4\mathcal{I}_{\tau_B}}{c^2}-\frac{\mathcal{I}_{\theta}}{r_B^2}\right)\cos\psi}\nonumber\\
    &\geq\min\left\{\frac{c^2}{4\mathcal{I}_{\tau_B}},\frac{r_B^2}{\mathcal{I}_{\theta}}\right\}.
    \label{eq:CRLBjLimit}
    \end{align}
    This implies that the CRLB does not vanish even when the delay estimation at the UE is arbitrarily accurate. This is because the BS-UE bi-static sensing alone cannot estimate the target position, due to the lack of angle estimation capability at UE.
\end{itemize}

Next, we consider a possible extension of the analytical results presented in Theorem~1. In particular, we may incorporate prior information about  $\theta$ and  account for the synchronization error between BS and UE.

\subsubsection{Incorporate the prior information on $\theta$}
When there is a single receiving antenna at the BS, i.e., $N_R=1$, we have $\mathcal{I}_{\theta}=0$ and hence $C_{\mathrm{mono}}(\mathbf{q})\rightarrow\infty$. In practice, under mono-static sensing with only delay  information available, the target's position can be confined to a ring or a semicircular arc centered at the BS, thank to the prior information of $\theta$, e.g.,  $\theta \in(-\pi, \pi]$ or $\theta \in(-\pi/2, \pi/2]$. If we assume that the target direction $\theta$ is uniformly distributed within $(-\pi/2, \pi/2]$. The impact of uniform prior information on the optimal estimation based measurement data is negligible \cite{kay1993fundamentals}, but it introduces
additional Fisher information, which revise $\mathcal{I}_{\theta}$ as
\begin{align}
\mathcal{\tilde{I}}_{\theta}\approx \frac{\pi^2 KM(N_R^2-1)N_R\Upsilon_B\cos^2(\theta)}{6}+\frac{12}{\pi^2}, \label{eq:FishTheta}
\end{align}
where the second term $\frac{12}{\pi^2}$ is negligible compared with first term if $N_R$ is relatively large, but it ensures $\mathbf{I}(\Theta_1)$ to be full rank and hence a finite value of $C_{\mathrm{mono}}(\mathbf{q})$ even when $N_R=1$.

\subsubsection{Account for the random synchronization error}The Fisher information of $\mathcal{I}_{\tau_U}$ in \eqref{eq:FishTauU0} is derived based on the assumption that the BS and  UE are perfectly synchronized. In practice, the synchronization error can be modelled as a random variable $\epsilon\sim \mathcal{N}(0,\sigma_s^2)$. Hence, the CRLB for estimating $\tau_U$ consists two parts, one from estimation error and the other from synchronization error, i.e,.
\begin{align}
\mathrm{var}(\hat{\tau}_U)\geq \frac{1}{\mathcal{I}_{\tau_U}}+\sigma_s^2.
\end{align}

Hence, the Fisher information of $\tau_U$ can be revised as
\begin{align}
\mathcal{\tilde I}_{\tau_U}=\frac{\mathcal{I}_{\tau_U}}{1+\mathcal{I}_{\tau_U}\sigma_s^2}. \label{eq:FisherTauU}
\end{align}

By substituting $\mathcal{I}_{\theta}$ and $\mathcal{I}_{\tau_U}$ with $\mathcal{\tilde{I}}_{\theta}$ and $\mathcal{\tilde I}_{\tau_U}$ in \eqref{eq:FishTheta} and \eqref{eq:FisherTauU}, we can obtain the CRLB for target localization in the hybrid mono- and bi-static sensing, in the presence of prior information and BS-UE synchronization error.

\subsection{CRLB for Target Velocity Estimation}
In this subsection, we extend the analysis to target velocity estimation. With either BS mono-static or BS-UE bi-static sensing, only the radial velocity $\mathbf{v}_B$ or $\mathbf{v}_U$ can be estimated from the measured Doppler $f_{D}^B$ and $f_{D}^U$. As shown in Fig.~\ref{F:VelCoordinate}, the radial velocity vector $\mathbf{v}_B$ is along the line connecting the BS and target, while the radial velocity vector $\mathbf{v}_U$ is aligned along the bisector of the bi-static angle \cite{6034675}, and their values are
\begin{align}
\|\mathbf v_B\|&=\left(-\frac{\mathbf{q}}{\|\mathbf q\|}\right)^T\cdot \mathbf{v}=\left[\begin{matrix}-\cos\theta \\ -\sin\theta\end{matrix}\right]^T\left[\begin{matrix}v_x \\ v_y\end{matrix}\right], \label{eq:radialVelB}\\
\|\mathbf v_U\|&=\left(-\frac{\mathbf{q}}{\|\mathbf q\|}-\frac{\mathbf{q}-\mathbf{q}_U}{\|\mathbf{q}-\mathbf{q}_U\|}\right)^T\cdot \mathbf{v}\nonumber\\
&=\left[\begin{matrix}\cos\theta-\cos\phi \\ \sin\theta-\sin\phi\end{matrix}\right]^T\left[\begin{matrix}v_x \\ v_y\end{matrix}\right]. \label{eq:radialVelU}
\end{align}

\begin{figure}[htb]
\centering
\includegraphics[width=0.3\textwidth]{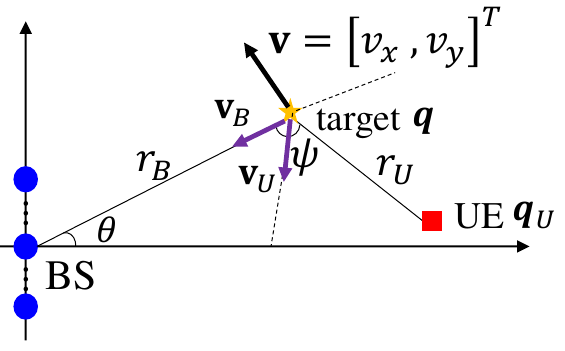}
\caption{Target velocity and measured Doppler by BS and UE.}
\label{F:VelCoordinate}
\end{figure}

With hybrid mono- and bi-static sensing, the velocity vector can be estimated from the measured Doppler at BS and UE. Let $\Theta_2=\{f_D^B, f_{D}^U\}\subset\Theta$, and it is related with velocity vector as
\begin{align}
\Theta_2=\mathbf{g}_2(\mathbf{v})=\left[\begin{matrix}\frac{2}{\lambda_c}\|\mathbf v_B\| \\ \frac{1}{\lambda_c}\|\mathbf v_U\| \end{matrix}\right].\label{eq:g2}
\end{align}
Since the estimation of $f_D^B$ and $f_D^U$ are independent, the EFIM for $\Theta_2$ can be directly obtained from the EFIM for each element,  i.e.,
\begin{align}
 \mathbf{I}(\Theta_2)=\left[\begin{matrix}\left[\mathbf{I}(\Theta_B)\right]_{44} & 0 \\ 0 & \left[\mathbf{I}(\Theta_U)\right]_{33}\end{matrix}\right]=\left[\begin{matrix}\mathcal{I}_{f_D^B} & 0 \\ 0 & \mathcal{I}_{f_D^U}\end{matrix}\right],
\end{align}
where $\mathcal{I}_{f_D^B}$ and $\mathcal{I}_{f_D^U}$ are the Fisher Information for  $f_D^B$ and $f_D^U$, respectively. According to \cite{pucci2025}, we have
\begin{align}
\mathcal{I}_{f_D^B}&=\frac{2\pi^2T_s^2KM(M^2-1)N_R\Upsilon_B}{3}\label{eq:FisherfDB},\\
\mathcal{I}_{f_D^U}&=\frac{2\pi^2T_s^2KM(M^2-1)\Upsilon_U}{3}\label{eq:FisherfDU}.
\end{align}
 The Jacobian of the transformation from $(f_D^B,f_D^U)$ to target velocity $\mathbf{v}=[v_x,v_y]^T$ is derived from \eqref{eq:radialVelB},\eqref{eq:radialVelU} and \eqref{eq:g2} as
\begin{align}
&\mathbf{J}_{\mathbf v}=\left[\begin{matrix}\frac{\partial f_{D}^B}{\partial v_x} & \frac{\partial f_{D}^B}{\partial v_y} \\
\frac{\partial f_{D}^U}{\partial v_x} & \frac{\partial f_{D}^U}{\partial v_y} \end{matrix}\right]=\frac{1}{\lambda_c}\left[\begin{matrix}-2\cos\theta & -2\sin\theta \\
\cos\theta-\cos\phi & \sin\theta-\sin\phi \end{matrix}\right]
\end{align}

Hence, the CRLB for velocity estimation in hybrid mono- and bi-static  sensing mode is
\begin{align}
C_{\mathbf{v}}(\mathbf{q},\mathbf{q}_U)&=\mathrm{Tr}\left((\mathbf{J}_{\mathbf v}^T\mathbf{I}(\Theta_2)\mathbf{J}_{\mathbf v})^{-1}\right)\nonumber\\
&=\lambda_c^2\cdot\frac{2\mathcal{I}_{f_D^B}+(1-\cos(\theta-\phi))\mathcal{I}_{f_D^U}}{2\mathcal{I}_{f_D^B}\mathcal{I}_{f_D^U}\sin^2(\theta-\phi)}\nonumber\\
&=\frac{\lambda_c^2}{\mathcal{I}_{f_D^U}\sin^2\psi}+\frac{\lambda_c^2}{2\mathcal{I}_{f_D^B}(1-\cos\psi)}. \label{eq:CRLBV}
\end{align}

Since $\mathcal{I}_{f_D^B}$ and $\mathcal{I}_{f_D^U}$ are only related with the target and UE positions through the SNR, we can express the CRLB of velocity estimation as the geometric relationship among BS, target and UE.
\begin{theorem}
The CRLB for velocity estimation with hybrid mono- and bi-static OFDM-ISAC is  given by
\begin{align}
C_{\mathbf{v}}(\mathbf{q},\mathbf{q}_U)&=c_5\frac{r_B^2 r_U^2}{\sin^2\psi}+c_6\frac{r_B^4}{1-\cos\psi}\label{eq:CRLBvelocity}
\end{align}
where
\begin{align*}
c_5&=\frac{3\lambda_c^2}{2\pi^2T_s^2KM(M^2-1)\bar{\Upsilon}_U},\\
c_6&=\frac{3\lambda_c^2}{4\pi^2T_s^2KM(M^2-1)N_R\bar{\Upsilon}_B},
\end{align*}
 are constants independent of target and UE positions
\end{theorem}

In practice, the velocity estimation error is usually measured by the VEB, which is defined as the square root of the CRLB for $\mathbf{v}$ as
\begin{align}
VEB(\mathbf{q},\mathbf{q}_U)=\sqrt{C_{\mathbf v}(\mathbf{q},\mathbf{q}_U)}.
\end{align}

\begin{example}
Under the same set of parameters as Example 1, Fig.~\ref{F:VEBresults}(a) shows the VEB for the potential targets with $0\leq x\leq 300\ m$ and $-150\leq y\leq 150\ m$ and the fixed cooperating UE position at $\mathbf{q}_U=[300,0]^T$. It is observed that a PEB lower than $2$ m/s can be achieved for majority part of the region, and the PEB is lower than $2$m/s when the targets are close to the BS or UE.  Fig.~\ref{F:VEBresults}(b) shows the VEB  of the target at $\mathbf{q}=[200,50]^T$  with the measurement results from the potential UE with the area. Similar to  PEP, it is observed that the PEB is smaller if the UE is close to the target, due to the larger SNR of the bi-static sensing signal.

\begin{figure}[htb]
\centering
\begin{subfigure}{0.24\textwidth}
\centering
\includegraphics[width=\textwidth]{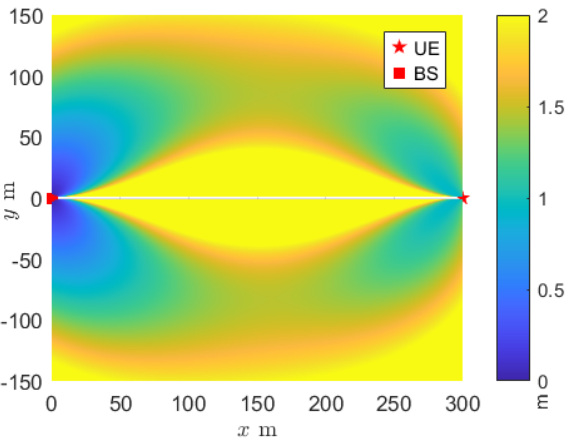}
\caption{Fixed UE at $[300,0]^T$}
\end{subfigure}
\begin{subfigure}{0.24\textwidth}
\centering
\includegraphics[width=\textwidth]{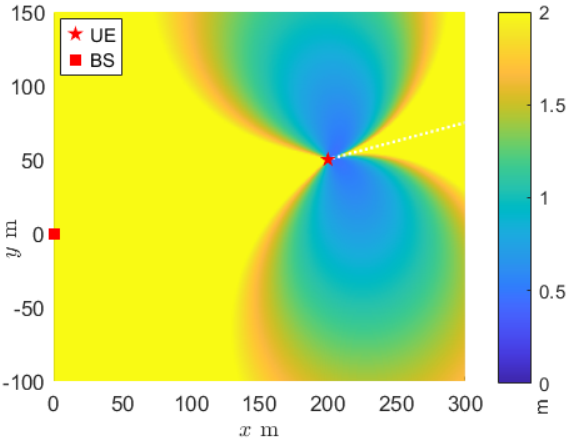}
\caption{Fixed target at $[200,50]^T$}
\end{subfigure}
\caption{The VEB for hybrid mono- and bi-static  sensing for given UE and target positions, respectively.}
\label{F:VEBresults}
\end{figure}
\end{example}

From the closed-form expression of the CRLB for  velocity estimation in \eqref{eq:CRLBvelocity}, we can obtain the following insights:
\begin{itemize}
\item The velocity estimation error bound $C_{\mathrm{v}}(\mathbf{q},\mathbf{q}_U)$  only depends on the sensing resource used (i.e., number of subcarriers $K$, number of OFDM symbols $M$, etc.) and the positions of target and UE, while independent  of the actually velocity of the target $\mathbf{v}$.
\item With either $\mathcal{I}_{f_D^U}\rightarrow 0$ or $\mathcal{I}_{f_D^B}\rightarrow 0$, we have $C_{\mathbf v}(\mathbf{q},\mathbf{q}_U)\rightarrow \infty$. This implies that either pure mono-static or bi-static sensing cannot estimate the target velocity.
\item When $\psi=0$ or $\psi=\pi$, $C_{\mathbf v}(\mathbf{q},\mathbf{q}_U)\rightarrow\infty$. This implies that accurate velocity estimation is not feasible if the BS, target and UE are on the same line.
\end{itemize}

\subsection{Summary of the Analytical Results}
The CRLB expressions for target position and velocity estimation with different sensing modes in OFDM-ISAC are summarized in Table~\ref{tb:summary} for convenience. By comparing the similarity between the CRLB for velocity estimation  and that for localization under the single-antenna assumption, we conclude that the optimal geometric configuration for achieving the best velocity estimation also provides a favorable geometry for accurate target localization.

\begin{table*}[htb]
\caption{A summary of CRLB for position and velocity estimation with various sensing modes}
\centering
\label{tb:summary}
\begin{tabular}{|c|c|c|c|c|}
\hline
{Sensing Mode} & \textbf{BS mono-static} & \textbf{BS-UE bi-static} & \textbf{Hybrid mono- and bi-static}  \\
\hline
\textbf{Position CRLB}  & $c_1\frac{r_B^6}{\cos^2\theta}+c_2r_B^4$ & $\infty$ & $C_{h}(\mathbf{q},\mathbf{q}_U)$ in \eqref{eq:CRLBjoint_polar} \\
\hline
\textbf{Position CRLB } ($N_R=1$)   & $\infty$ & $\infty$ & $c_3\frac{r_B^2r_U^2}{\sin^2\psi}+c_4\frac{r_B^4}{1-\cos\psi}$ \\
\hline
\textbf{Velocity CRLB }   & $\infty$ & $\infty$ & $c_5\frac{r_B^2 r_U^2}{\sin^2\psi}+c_6\frac{r_B^4}{1-\cos\psi}$ \\
\hline
\end{tabular}
\end{table*}

Furthermore, by differentiating the CRLB expressions with respect to $\psi$, we can find the optimal BS-target-UE geometry that minimize the CRLB for the hybrid mono- and bi-static sensing.
\begin{lemma}
The CRLB of target position/velocity estimation based on the delay/Doppler estimation at BS and UE is minimized when BS-target-UE forms an angle equal to
    \begin{align}
        \psi^*=\arccos\left(2\sqrt{\rho(\rho+1)}-2\rho-1\right),
    \end{align}
    where $\rho=\mathcal{I}_{\tau_B}/\mathcal{I}_{\tau_U}=\mathcal{I}_{f_D^B}/\mathcal{I}_{f_D^U}=N_Rr_U^2/r_B^2$ is the ratio between the Fisher information on delay/Doppler estimated at the UE and BS, respectively.
\end{lemma}
In general, the optimal angle $\psi^*$ is within the range $(\frac{\pi}{2},\pi)$.
    When the delay/Doppler estimation at the BS is much more accurate than that estimated by UE, i.e., $\rho\gg 1$, the velocity estimation error $C_{\mathbf v}(\mathbf{q},\mathbf{q}_U)$ is minimized when  $\psi\rightarrow\frac{\pi}{2}$.
%

\section{Verification of the Analytical Results} \label{sec:verify}
The analysis in the preceding section quantifies the reduction of PEB by fusing the delay measurement from the UE as a function of the target and UE positions. This section aims to verify the analytical results using a practical  parameter estimation algorithm. In particular, the target localization based on BS mono-static sensing can be formulated as
\begin{align}
\textbf{P1:}\quad \mathbf{\hat q}_{\mathrm{mono}}=\arg\min_{\mathbf{q}}\sum_{m=0}^{M-1}\sum_{k=0}^{K-1}\left\|\frac{\mathbf{y}_{km}^B}{d_{km}}-\mathbf{f}_B(\mathbf{q})\right\|^2
\end{align}
where $\mathbf{f}_B(\mathbf{q})=\mathbf{b}(\theta)\alpha_B\beta e^{-j2\pi k\Delta f\tau_B}e^{j2\pi mT_sf_{D}^{B}}$, with $\tau_B=2\|\mathbf{q}\|/c$, $\theta=\arctan(\mathbf{q}(2)/\mathbf{q}(1))$.

Solving \textbf{P1} directly is challenging, since the complex path gain $\alpha_B$ and the Doppler shift $f_D^B$ in $\mathbf{f}_B(\mathbf{q})$ are also unknown variables. To this end, we adopt a suboptimal, yet efficient, algorithm for solving \textbf{P1}, which contains the following two steps:
\begin{itemize}
\item{Estimate the parameter set $\Theta_B=\{\hat{\tau}_B,\hat{f}_D^B,\hat{\theta}\}$ based on the received signal $\mathbf{y}_{km}^B$.  }
\item{Estimate the target position $\mathbf{\hat q}_{\mathrm{mono}}$ from $(\hat\tau_B,\hat{\theta})$}.
\end{itemize}
Since this paper focuses on the performance gain achieved by fusing the UE results, rather than optimizing a single parameter estimation, we adopt sequential estimation for the parameters in $\Theta_B$ based on FFT for computational efficiency.  The procedures of sequential parameter estimation follow the description in \cite{DaiQiang}, which are summarized in Algorithm~\ref{alg1}. Next, since $\mathbf{q}$ can be written as deterministic function of $\tau_B$ and $\theta$, we directly compute
\begin{align}
\mathbf{\hat q}_{\mathrm{mono}}=\frac{c\hat{\tau}_B}{2}\left[\cos\hat{\theta}, \sin\hat{\theta}\right]^T.
\end{align}

\begin{algorithm}[htb]
	\caption{OFDM-ISAC Sequential Parameter Estimation}
	\begin{algorithmic}[1]
		\STATE{\textbf{Input} : Received signal after data removal, i.e.,  $\mathbf{\tilde y}_{km}^B={\mathbf{y}_{km}^B}/{d_{km}}, m=0,...,M-1, k=0,...,K-1$}.
        \STATE{\textbf{Initialize}: FFT over-sampling factor in angular, Doppler and delay dimension $L_{\theta}$, $L_{D}$, $L_{\tau}$.}
        \STATE{Obtain $\hat{\theta}$ by performing $L_{\theta}N_R$-point FFT over each vector $\mathbf{\tilde y}_{km}^B$ and sum over all the $K$ subcarriers and $M$ OFDM symbols. }
        \STATE{Perform beamforming over signal received at different antennas based on $\hat{\theta}$, i.e., $\tilde{y}_{km}^B=\mathbf{b}(\hat{\theta})^H\mathbf{\tilde y}_{km}^B$}
        \STATE{Obtain $\hat{\tau}_B$ by performing $L_{\tau}M$ point FFT over $\{\sum_{m=0}^{M-1}\tilde{y}_{km}^B, k=0,...,K-1\}$.}\label{line:tau}
        \STATE{Obtain $\hat{f}_D^B$ by performing $L_{D}M$ point IFFT over the vector $\{\sum_{k=0}^{K-1}e^{-jk\Delta f\hat{\tau}_B}\tilde{y}_{km}^B, m=0,...,M-1\}$ .}\label{line:fD}
        \STATE{\textbf{Output}: $\hat{\theta}$, $\hat{\tau}_B$, $\hat{f}_D^B$.}
	\end{algorithmic}
	\label{alg1}
\end{algorithm}

For the hybrid mono- and bi-static sensing, the target localization is formulated as
\begin{align}
\textbf{P2:} \arg\min_{\mathbf{q}}\sum_{m=0}^{M-1}\sum_{k=0}^{K-1}\left\|\frac{\mathbf{y}_{km}^B}{d_{km}}-\mathbf{f}_B(\mathbf{q})\right\|^2+\left|\frac{y_{km}^u}{d_{km}}-f_U(\mathbf{q})\right|^2,
\end{align}
where $f_U(\mathbf{q})=\alpha_U\beta e^{-j2\pi k\Delta f\tau_U}e^{j2\pi mT_sf_{D}^{U}}$. Similarly, we first obtain $(\hat\tau_U, \hat{f}_D^U)$ using the standard OFDM radar receiver processing, i.e., line \ref{line:tau}-\ref{line:fD} in Algorithm~\ref{alg1}. Then, the target localization is obtained by solving the weighted MSE minimization problem \textbf{P3} in \eqref{eq:P3}, where the weighting factor $\mathrm{var}(\hat\tau_B)=1/\mathcal{I}_{\tau_B}$, $\mathrm{var}(\hat\theta)=1/\mathcal{I}_{\theta}$ and $\mathrm{var}(\hat\tau_U)=1/\mathcal{I}_{\tau_U}$ are the expected variance for estimation of $\hat{\tau}_B$, $\hat{\theta}$ and $\hat{\tau}_U$. Although the closed-form solution to \textbf{P3} is not available,  efficient iterative algorithms for minimizing the weighted MSE have been well studied, and we adopt the classical Gauss-Newton algorithm for its simplicity.
\begin{figure*}
\begin{align}
\textbf{P3:}\quad
\mathbf{\hat q}_h=\arg\min_{\mathbf{q}}\left[\frac{(\hat{\tau}_B-2\|\mathbf{q}\|/c)^2}{\mathrm{var}(\hat\tau_B)}+\frac{(\hat{\tau}_U-(\|\mathbf{q}\|+\|\mathbf{q}-\mathbf{q}_U\|)/c)^2}{\mathrm{var}(\hat\tau_U)}+\frac{(\hat{\theta}-\arctan(\mathbf{q}(2)/\mathbf{q}(1))^2}{\mathrm{var}(\hat\theta)}\right].\label{eq:P3}
\end{align}
\end{figure*}

For target velocity estimation, we first calculate the directional vector  based on the estimated target position $\mathbf{\hat{q}}_h$, i.e.,
\begin{align}
\mathbf{\hat u}_t=-\frac{\mathbf{\hat{q}}_h}{\|\mathbf{\hat{q}}_h\|}, \quad \mathbf{\hat u}_r=-\frac{\mathbf{q}_U-\mathbf{\hat{q}}_h}{\|\mathbf{q}_U-\mathbf{\hat{q}}_h\|}.\label{eq:unitVec}
\end{align}

Then, according in \eqref{eq:radialVelB}-\eqref{eq:radialVelU}, the estimated velocity can be directly computed from the estimated Doppler $\hat{f}_D^B$ and $\hat{f}_D^U$ as
\begin{align}
\mathbf{\hat{v}}=g_2^{-1}(\Theta_2)=\left[\begin{matrix}\mathbf{\hat u}_t^T \\ (\mathbf{\hat u}_r-\mathbf{\hat u}_t)^T\end{matrix}\right]^{-1}
\left[\begin{matrix}\hat{f}_D^B \\ \hat{f}_D^U\end{matrix}\right].\label{eq:estiVel}
\end{align}

\begin{example}
Based on the OFDM setting as Example 1, we compare the position estimation error using the practical algorithms described above with the analytical results given by \eqref{eq:PEBmono} and \eqref{eq:PEBjoint}, for a target located at $\mathbf{q}=[200,50]^T$. The UE position is $\mathbf{q}_U=[0,300]^T$. The position estimation error is plotted versus the receive SNR at the BS in Fig.~\ref{F:PEBcompare}. According to the SNR expression in \eqref{eq:SNRm} and \eqref{eq:SNRu}, the receive SNR at the UE is proportional to that at the BS, with a scaling factor $\frac{\kappa_Ur_B^2}{\kappa_Br_U^2}$.
To ensure that the single-parameter estimates are sufficiently close to the CRLB, we use a large oversampling factor for the FFT\footnote{It has been noted in \cite{Braun2014} that FFT with sufficiently large oversampling is equivalent to maximum likelihood (ML) estimation for OFDM radar when there is only a single target. },  i.e., $L_{\theta}=2048$ and $L_{\tau}=1024$. The results are averaged over 500 Monte Carlo simulations.
It is observed from Fig.~\ref{F:PEBcompare} that the simulation results approach the analytical PEB for both BS mono-static and hybrid sensing modes, when the receive SNR gets larger.
\begin{figure}[htb]
\centering
\includegraphics[width=0.4\textwidth]{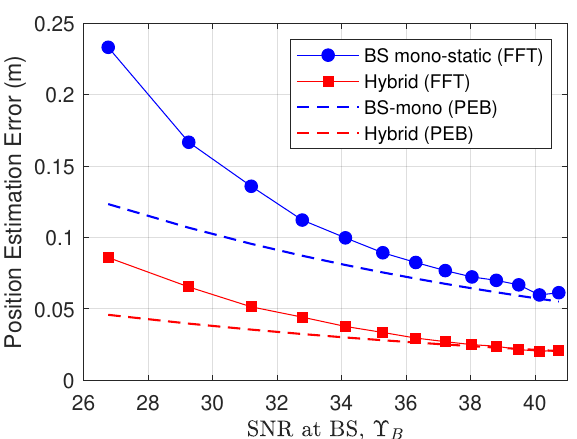}
\caption{Compare the position estimation error with PEB.}
\label{F:PEBcompare}
\end{figure}

Next, consider the target has a velocity $\mathbf{v}=[20, 0]^T$ m/s. The velocity is estimated based \eqref{eq:unitVec}-\eqref{eq:estiVel} after the position is estimated.  Since the number of OFDM symbols we choose is rather small, a large oversampling factor $L_D=4096$ is used so that the FFT estimation can approximate the ML estimation. It is observed from Fig.~\ref{F:VEBcompare} that, as the transmit power becomes large, the estimation error approaches the derived VEB, which is the square root of the CRLB in \eqref{eq:CRLBV}. The simulation slightly surpasses the analysis when SNR is large. This is because the analysis is derived based on the assumption of separate position and velocity estimation, while the coherent summation with estimated $\hat{\tau}_B$ is adopted in line~\ref{line:fD} of Algorithm~\ref{alg1}.
\begin{figure}[htb]
\centering
\includegraphics[width=0.4\textwidth]{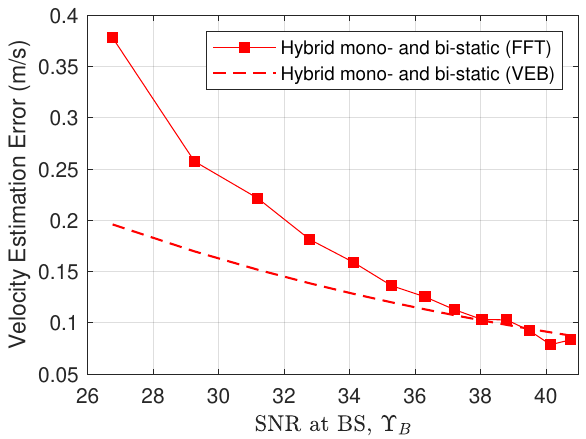}
\caption{Compare the velocity estimation error with VEB.}
\label{F:VEBcompare}
\end{figure}
\end{example}

\section{Coverage Analysis for Hybrid Sensing}\label{sec:utilization}
In this section, we consider the utilization of analytical results in network design. First, the network coverage analysis with selected BS and UE is presented, based on which we can optimize the UE position to maximize the network coverage. Second, considering random UE distribution, the achievable PEB with optimal cooperative UE selection for a given target is presented as a function of UE density.

\subsection{Sensing Coverage Analysis}
The PEB and VEP derived in the preceding sections can be employed to evaluate the sensing coverage of  cellular network under different ISAC operating modes. For example, consider a target located at a random location $\mathbf{q}=[x,y]^T$. The target can be successfully localized in the BS mono-static sensing mode if the PEB at that position falls below a predefined threshold, i.e.,
\begin{align}
\sqrt{C_{\mathrm{mono}}(\mathbf{q})}\leq \gamma_{p}, \label{eq:LocCoverage}
\end{align}
where $\gamma_p$ is the required positioning accuracy. The coverage of the BS mono-static sensing mode is defined as all the potential target positions such that \eqref{eq:LocCoverage} is met, i.e,
\begin{align}
\mathcal{A}_{\mathrm{mono}}=\int\mathbf{1}(\sqrt{C_{\mathrm{mono}}(\mathbf{q})}\leq \gamma_{p}) d\mathbf{q},\label{eq:defCoverage}
\end{align}
where $\mathbf{1}(\cdot)$ is a binary indicator function. According to the analysis presented in the preceding sections,  the CRLB for target localization increases as the target moves farther away from the BS. Hence, the coverage area is the area enclosed by the curve defined by
\begin{align}
C_{\mathrm{mono}}(\mathbf{q})=\gamma_p^2. \label{eq:equalCurve}
\end{align}

Substituting the expression of $C_{\mathrm{mono}}(\mathbf{q})$ in \eqref{eq:CRLBm} into \eqref{eq:equalCurve}, we have
\begin{align}
c_1r_B^6+c_2r_B^4\cos^2(\theta)=\gamma_p^2\cos^2(\theta). \label{eq:areaPolar}
\end{align}
The curve defined by \eqref{eq:areaPolar} shares a similar shape as a lemniscate, but it does not belong to any classic lemniscate. A tight approximate closed-form expression for the area enclosed by \eqref{eq:areaPolar} is presented in the following lemma.

\begin{lemma}
Consider a given PEB threshold $\gamma_{p}$, the coverage of the BS mono-static sensing with $K$ subcarriers, $M$ OFDM symbols and $N_R$ receive antennas is approximately given by
\begin{align}
&\mathcal{A}_{\mathrm{mono}}\approx \nonumber\\
&2.6448\gamma_p^{\frac{2}{3}}\sqrt[3]{N_R(N_R^2-1)KM\bar{\Upsilon}_B}-\frac{0.0327c^2(N_R^2-1)}{\Delta f^2(K^2-1)},\label{eq:MonoCoverage}
\end{align}
\end{lemma}
\begin{proof}
The derivation is presented in Appendix~\ref{A:monoCoverage}.
\end{proof}

Since the second term in \eqref{eq:MonoCoverage} is just a small correction in the approximation, it is usually much smaller than the first term. Hence, we can conclude that the sensing coverage of BS mono-static sensing in OFDM-ISAC is proportional to $\gamma_p^{\frac{2}{3}}$, $(KM)^{\frac{1}{3}}$, and $N_R$  if $N_R\gg 1$.

Next, we consider the coverage area of the hybrid mono- and bi-static  sensing with a given UE position, which is defined as
\begin{align}
\mathcal{A}_h=\int\mathbf{1}\left(\sqrt{C_h(\mathbf{q},\mathbf{q}_U)}\leq \gamma_{p}\right) d\mathbf{q},\label{eq:defCoverage}
\end{align}
where $C_h(\mathbf{q},\mathbf{q}_U)$ is given in \eqref{eq:CRLBjoint_polar}. Deriving the closed-form expression for $\mathcal{A}_h$ is challenging. However, since the PEB can be efficiently calculated for given UE and target position, evaluating the coverage numerically is efficient.
\begin{example}
Under the parameter settings as Example 1, Fig.~\ref{F:PEBwUELoc}(a) illustrates the right-half coverage region of the BS mono-static sensing for a PEB threshold of $\gamma_p=1$ m. The full coverage area is symmetric about the origin and thus twice the area shown in Fig.~\ref{F:PEBwUELoc}(a), accounting for the left half as well. The accuracy of the analytical coverage expression in~\eqref{eq:MonoCoverage} for BS mono-static sensing is validated by the simulation results in Fig.~\ref{F:monoCov}, which demonstrate excellent agreement under various network configuration.

Besides, the coverage area is greatly enhanced with the hybrid mono- and bi-static sensing\footnote{Only the area on the right side of the BS is counted for both BS mono-static and hybrid sensing.},  as shown by the dashed lines in Fig.~\ref{F:monoCov}, with the UE located at $\mathbf{q}_U=[300,0]^T$. The enlarged coverage of the hybrid mono- and bi-static  sensing is negligible when $\gamma_p$ is small and it becomes larger when the PEB threshold gets larger.
\begin{figure}[htb]
\centering
\includegraphics[width=0.4\textwidth]{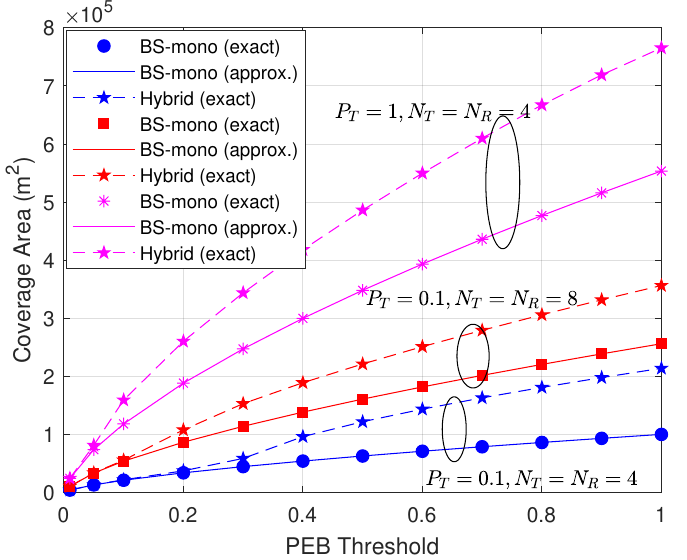}
\caption{Comparison of sensing coverage with various sensing modes.}
\label{F:monoCov}
\end{figure}
\end{example}


In the scenario where the UE serves as a dedicated device to assist the BS in sensing, the UE's deployment can be optimized based on coverage area. For instance, consider a cooperative UE located on the x-axis at position $\mathbf{q}_U=[x_U,0]^T$ and moving away from the BS, as illustrated in Fig.~\ref{F:optUELoc}(a). The sensing coverage of the hybrid mono- and bi-static sensing based on the cooperation of the BS and this UE is plotted in Fig.~\ref{F:optUELoc}(b). It is observed that the coverage initially increases and then decreases as the UE moves farther from the BS. Moreover, the optimal UE position $x_U^*$ shifts farther from the BS as the PEB threshold increases.

\begin{figure}[htb]
\centering
\begin{subfigure}{0.3\textwidth}
\centering
\includegraphics[width=\textwidth]{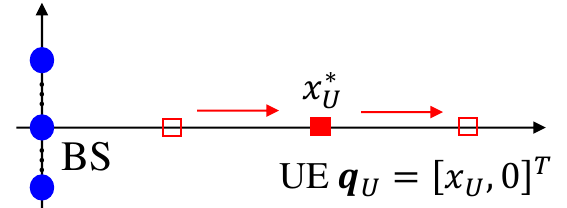}
\caption{Adjusting the UE position on $x$-axis.}
\end{subfigure}
\begin{subfigure}{0.35\textwidth}
\centering
\includegraphics[width=\textwidth]{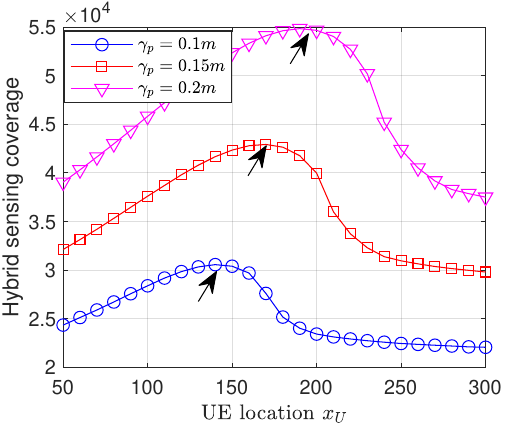}
\caption{Hybrid Sensing coverage}
\end{subfigure}
\caption{Sensing coverage with varying UE positions.}
\label{F:optUELoc}
\end{figure}


The coverage analysis can be readily extended to velocity estimation by replacing $C_h(\mathbf{q},\mathbf{q}_U)$ in \eqref{eq:defCoverage} with $C_\mathbf{v}(\mathbf{q},\mathbf{q}_U)$ in \eqref{eq:CRLBV}. Moreover, when both position and velocity estimation are of interest, the coverage region can be defined as the set of potential target positions for which both position and velocity estimates satisfy their respective accuracy requirements, i.e.,
\begin{align}
\mathcal{A}=\int\mathbf{1}\left(\sqrt{C_h(\mathbf{q},\mathbf{q}_U)}\leq \gamma_{p},\sqrt{C_\mathbf{v}(\mathbf{q},\mathbf{q}_U)}\leq \gamma_{v}\right) d\mathbf{q},\label{eq:defCoveragePV}
\end{align}
where $\gamma_v$ is the sensing threshold for velocity estimation. Although a closed-form expression for this coverage region is not available, numerical evaluation remains highly efficient, by using the closed-form expressions for the CRLB.


\subsection{PEB with UE Selection}
The key advantage of the proposed hybrid mono- and bi-static  sensing mode lies in its ability to leverage measurements from UE. By flexibly selecting an appropriate UE from the large pool of potential UEs in the network, the system can significantly enhance sensing performance. As shown in \eqref{eq:CRLBjoint_polar} and \eqref{eq:CRLBvelocity}, sensing accuracy improves when the selected UE is located close to the target and forms an angle with the BS and the target that is close to the optimal angle $\psi^*$. Assume that the BS knows the rough position of the target and hence it can always select the most proper UE to assist the sensing process. We analyze the distribution of PEB under random UE distribution in this section.

The distribution of UE is modeled by a homogeneous Poisson point (HPPP) process with density $\lambda$ per $m^2$. Considering a typical target at $\mathbf{q}$, if there is no UE to help, i.e., $\lambda=0$, the PEB is fixed and given by $PEB_{\mathrm{mono}}(\mathbf{q})$ in \eqref{eq:PEBmono}. When $\lambda>0$, a lower PEB can be achieved by fusing the measurement results from UE. Besides, if the UE is arbitrarily close to the target, the ultimate PEB lower bound can be obtained from \eqref{eq:CRLBjLimit}, which is
\begin{align}
PEB_h^{\mathrm{lim}}=\sqrt{\min\left\{\frac{c^2}{4\mathcal{I}_{\tau_B}},\frac{r_B^2}{\mathcal{I}_{\theta}}\right\}}.\label{eq:PEBcap}
\end{align}

For general PEB threshold  $PEB_h^{\mathrm{lim}}\leq\gamma_p\leq PEB_{\mathrm{mono}}(\mathbf{q})$, the cumulative distribution function (CDF)  for the PEB distribution is defined as
\begin{align}
F(\gamma_p,\mathbf{q})=\mathrm{Pr}(\exists \mathbf{q}_U: PEB_h(\mathbf{q},\mathbf{q}_U)\leq \gamma_p).\label{eq:probUE}
\end{align}

To derive the probability of UE existence in \eqref{eq:probUE}, we need to find the condition of UE positions for given target position $\mathbf{q}$. Based on the simulation results shown in Fig.~\ref{F:PEBwTargetLoc}, the admitted area of the UE position for given PEB threshold can be approximated by a lemniscate centered at the target. Mathematically, we solve for $\mathbf{q}_U$ with given $\mathbf{q}$ such that
\begin{align}
C_h(\mathbf{q},\mathbf{q}_U)=\gamma_p^2.  \label{eq:solveru}
\end{align}
To extract those parameters related with $\mathbf{q}_U$, denote $\mathcal{I}_{\tau_U}=\tilde{\mathcal{I}}_{\tau_U}/r_U^2$ and hence we have
\begin{align}
\tilde{\mathcal{I}}_{\tau_U}=\frac{2\pi^2\Delta f^2MK(K^2-1)N_R \hat{\Upsilon}_U}{3},\label{eq:ItauUtilde}
\end{align}
and $\hat{\Upsilon}_U=\Upsilon_U/r_U^2$ is the SNR coefficient at UE by excluding the $r_U$ factor.

Substituting \eqref{eq:CRLBjoint_polar} and \eqref{eq:ItauUtilde} into \eqref{eq:solveru}, we can solve the value of $r_U$ for given target position $\mathbf{q}$ as
\begin{align}
r_U^2=&\frac{(4r_B^2\mathcal{I}_{\tau_B}+c^2\mathcal{I}_{\theta})\gamma_p^2(1+\cos\psi)+2c^2\left(\gamma_p^2\mathcal{I}_\theta-r_B^2\right)}
{4r_B^2\mathcal{I}_{\tau_B}+c^2\mathcal{I}_{\theta}-4\gamma_p^2\mathcal{I}_{\tau_B}\mathcal{I}_{\theta}}\nonumber\\
& \cdot\frac{\tilde{\mathcal{I}}_{\tau_U}(1-\cos\psi)}{c^2}, \label{eq:ruExplicit}
\end{align}

If a UE is present within the region enclosed by the curve defined in~\eqref{eq:ruExplicit}, the BS-UE joint sensing mode can achieve a PEB no greater than $\gamma_p$. Under the HPPP model, the probability of UE presence in a given region depends solely on the area of that region, not on its specific shape. In particular, the probability that no UE is located within the region, known as the void probability, is given by
\begin{align}
\Pr(\textnormal{no UE in region } \mathcal{R}_U)=e^{-\lambda |\mathcal R_U|}. \label{eq:void}
\end{align}

Next, we need to evaluate the area enclosed by \eqref{eq:ruExplicit}. Although it is not classic lemniscate, its area can be explicitly calculated based on the standard area equation in polar coordinates. Depending on the value of $\gamma_p$, the admitted area is either enclosed by a full loop (i.e., $\psi\in [0,2\pi]$) or a partial loop (i.e., limited range of $\psi$ that ensures $r_U^2\geq 0$), with the explicit expression given by the following lemma.
\begin{lemma}
For OFDM-ISAC with hybrid mono- and bi-static sensing, a localization variance for a target at  $\mathbf{q}$ lower than $\gamma_p^2$ can be achieved if there exists a cooperative UE within the region enclosed by the curve in \eqref{eq:ruExplicit}, and the area of the enclosed region is given by
\begin{align}
|\mathcal R_U|=
\begin{cases}
f_1(\gamma_p^2,\mathbf{q})\left(P_1-\frac{P_2}{2}\right), & \gamma_p^2\geq \frac{r_B^2}{\mathcal{I}_\theta}\\
f_1(\gamma_p^2,\mathbf{q})f_2(P_1,P_2), & \gamma_p^2< \frac{r_B^2}{\mathcal{I}_\theta}
\end{cases} \label{eq:AreaUE}
\end{align}
where
\begin{align}
&f_1(\gamma_p^2,\mathbf{q})=\frac{\pi\tilde{\mathcal{I}}_{\tau_U}/c^2}{C_{\mathrm{mono}}(\mathbf{q})-\gamma_p^2}\\
&P_1=\gamma_p^2\left(\frac{r_B^2}{\mathcal{I}_\theta}+\frac{c^2}{4\mathcal{I}_{\tau_B}}\right)-\frac{r_B^2c^2}{2\mathcal{I}_{\tau_B}\mathcal{I}_\theta}\label{eq:P1}\\
&P_2=\gamma_p^2\left(\frac{r_B^2}{\mathcal{I}_\theta}-\frac{c^2}{4\mathcal{I}_{\tau_B}}\right).\label{eq:P2}\\
&f_2(P_1,P_2)=\left(P_1-\frac{P_2}{2}\right)\arccos\left(-\frac{P_1}{P_2}\right)\nonumber\\
&\quad \quad \quad \quad \quad +\left(P_2-\frac{P_1}{2}\right)\sqrt{1-\frac{P_1^2}{P_2^2}}
\end{align}
\end{lemma}
\begin{proof}
The derivation of \eqref{eq:AreaUE} follows by standard polar area integration with \eqref{eq:ruExplicit}, and hence is omitted for brevity.
\end{proof}

%


Substituting \eqref{eq:AreaUE} and \eqref{eq:void} into \eqref{eq:probUE}, we obtain the CDF of an arbitrary target in the network under BS-UE joint sensing with proper UE selection.
\begin{lemma}
Consider a network where the UEs are  distributed according to HPPP with density $\lambda$, the PEB for locating a target at $\mathbf{q}$ with hybrid mono- and bi-static  sensing via the best UE selection, denoted by $\gamma_p$, has CDF given by
\begin{align}
F(\gamma_p,\mathbf{q})=
\begin{cases}
0 & \gamma_p\leq PEB_h^{\mathrm{lim}}\\
1 &  \gamma_p\geq PEB_m(\mathbf{q})\\
1-e^{-\lambda |\mathcal R_U|}, & \textnormal{otherwise},
\end{cases}\label{eq:CDF}
\end{align}
where $|\mathcal R_U|$ is given by \eqref{eq:AreaUE}.
\end{lemma}

\begin{example}
To interpret the results in \eqref{eq:CDF}, we consider the sensing of a target at $\mathbf{q}=[200,50]^T$ . The PEB is 0.9134m under BS mono-static setting. If there is a UE arbitrarily close to the target, the hybrid mono- and bi-static  sensing achieves the PEB 0.1081m. For random UE distribution, the CDF of achievable PEB for varying UE density is shown in Fig.~\ref{F:PEBcdf}. When the UE density is high, it is very likely to find a cooperating UE close to the target, and hence the PEB is smaller.
\begin{figure}[htb]
\centering
\includegraphics[width=0.4\textwidth]{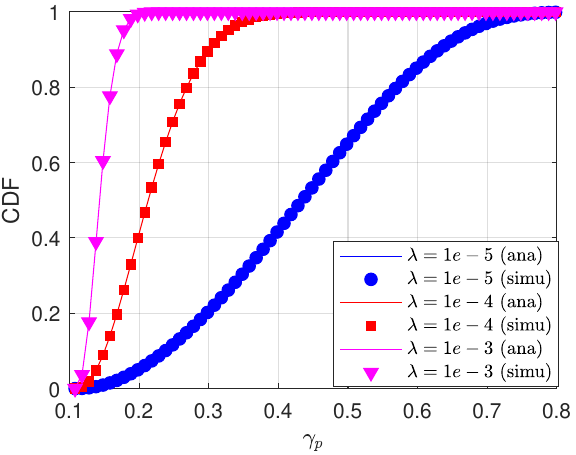}
\caption{The CDF of PEB for target at $\mathbf{q}=[200,50]^T$ with various UE density.}
\label{F:PEBcdf}
\end{figure}
\end{example}

\section{Conclusion}\label{sec:conclusion}
In this paper, we provided the closed-form expressions of CRLB for target position and velocity estimation with hybrid mono- and bi-static sensing, enabled by BS-UE cooperation, which were derived as explicit functions of the target and UE positions and validated using practical sensing algorithms. The analysis revealed that the performance enhancement is most pronounced when the BS-target-UE geometry approaches the optimal bi-static angle, and it vanishes entirely for both PEB and VEB when the three nodes are collinear. These simple expressions provide a tractable way to assess the intrinsic sensing capability of cellular networks with ISAC functionality, e.g., enabling evaluations of sensing coverage and the impact of user density on sensing performance. The analytical results also provide valuable guidelines for practical network design, such as UE selection for cooperative sensing.

For simplicity, the analysis assumed separate estimation of position and velocity. Future work could extend this framework to joint position-velocity estimation and multi-target scenarios. Additionally, fusing measurements from multiple UEs and considering the UE with 3D location, e.g., UAVs, present other promising directions for investigation.

\bibliographystyle{ieeetr}
\bibliography{IEEEabrv,ISACTutorial}
\appendices
\section{CRLB for Localization with Hybrid Sensing}\label{A:positionCRLB}
With the EFIM being a diagonal matrix,  the matrix product $\mathbf{J}_{\mathbf{q}}^T\mathbf{I}(\Theta_1)\mathbf{J}_{\mathbf{q}}$ is a symmetric matrix of size $2\times 2$, which can be represented as
\begin{align}
\mathbf{J}_{\mathbf{q}}^T\mathbf{I}(\Theta_1)\mathbf{J}_{\mathbf{q}}=\left[
\begin{matrix}
A & B\\
B & C
\end{matrix}\right],
\end{align}
where
\begin{align}
A&={J}_{11}^2\mathcal{I}_{\tau_B}+{J}_{21}^2\mathcal{I}_{\theta}+{J}_{31}^2\mathcal{I}_{\tau_U} \label{eq:A}\\
B&=J_{11}{J}_{12}\mathcal{I}_{\tau_B}+{J}_{21}{J}_{22}\mathcal{I}_{\theta}+{J}_{31}{J}_{32}\mathcal{I}_{\tau_U}\label{eq:B}\\
C&={J}_{12}^2\mathcal{I}_{\tau_B}+{J}_{22}^2\mathcal{I}_{\theta}+{J}_{32}^2\mathcal{I}_{\tau_U},\label{eq:C}
\end{align}
with $J_{ij}$ being the $(i,j)$th element of Jacobian matrix $\mathbf{J}(\mathbf{q})$. Then, the trace of its inverse is given by
\begin{align}
\mathrm{Tr}\left((\mathbf{J}_{\mathbf{q}}^T\mathbf{I}(\Theta_1)\mathbf{J}_{\mathbf{q}})^{-1}\right)=\frac{A+C}{AC-B^2}. \label{eq:traceABC}
\end{align}
Substituting \eqref{eq:A}-\eqref{eq:C} into \eqref{eq:traceABC} gives the explicit expression of the CRLB in terms of the coordinates of target and UE.
To further simplify the expression, we convert the expression to the polar coordinates representation, by substituting
\begin{align*}
&x=r_B\cos(\theta), \quad \quad \quad y=r_B\sin(\theta)\\
&x_U=x+r_U\cos(\phi), \quad y_U=y+r_U\sin(\phi)
\end{align*}
where $\phi$ is the direction of the UE with respect to the target. Then, the CRLB for hybrid position estimation becomes \eqref{eq:CRLBjoint_polar}, with $\psi\triangleq\pi-(\theta-\phi)=\arccos\left(\frac{\mathbf{q}^T(\mathbf{q}_U-\mathbf q)}{\|\mathbf{q}\|\|\mathbf{q}_U\|}\right)$ being the angle formed by BS-target-UE.

\section{Derive the Coverage of BS Mono-static Sensing}\label{A:monoCoverage}
With the BS antenna array on the $y$-axis, the coverage area of BS mono-static sensing has two symmetric lobes on two sides of $y$-axis. We find area of one lobe using the standard polar area formula and then multiply the results by 2 to get the total area. Mathematically, we have
\begin{align}
\mathcal{A}_m=2\cdot\frac{1}{2}\int_{-\pi/2}^{\pi/2}(r_B(\theta))^2d\theta.  \label{eq:areaPolar2}
\end{align}

To solve for $r_B(\theta)$, we let $u=r_B^2$ and \eqref{eq:areaPolar} reduces to
\begin{align}
u^3+\frac{c_2}{c_1}u^2\cos^2(\theta)=\frac{\gamma_p^2}{c_1}\cos^2(\theta). \label{eq:uEq}
\end{align}

Solving \eqref{eq:uEq} is challenging in general. However, when $c_2/c_1\rightarrow 0$, \eqref{eq:uEq} can be easily solved with solution
\begin{align}
u_0=\sqrt[3]{\gamma_p^2\cos^2(\theta)/c_1}=\gamma_p^{\frac{2}{3}}c_1^{-\frac{1}{3}}\cos^{\frac{2}{3}}(\theta).
\end{align}

In practice, the expected target localization error should be much larger than the distance estimation error, since the former is much more stringent. Hence, we have $\gamma_p\gg c\sqrt{\mathrm{var}(\hat\tau_B)}/2$, where $\sqrt{\mathrm{var}(\hat\tau_B)}$ is the standard deviation for estimating the delay $\tau_B$. Further consider that $\mathcal{I}_{\tau_B}=\frac{1}{\mathrm{var}(\hat\tau_B)}$ and  $c_2=\frac{c^2}{4r_B^4\mathcal{I}_{\tau_B}}$, we have $\gamma_p^2\gg c_2r_B^4=c_2u^2$, which is  equivalent to $\gamma_p^2/c_1\gg c_2/c_1$.  Hence, the small value of $\frac{c_2}{c_1}u^2\cos^2(\theta)$ can be considered as a perturbation of the solution. Denote by $u_1$ the correction introduced, with $u_1\ll u_0$. Substituting the solution $u=u_0+u_1$ into \eqref{eq:uEq}, we can solve for $u_1$ as
\begin{align}
u_1=-\frac{c_2}{3c_1}\cos^2(\theta).
\end{align}
Hence, the solution to \eqref{eq:uEq} can be approximated as
\begin{align}
u\approx u_0+u_1=\gamma_p^{\frac{2}{3}}c_1^{-\frac{1}{3}}\cos^{\frac{2}{3}}(\theta)-\frac{c_2}{3c_1}\cos^2(\theta). \label{eq:approxu}
\end{align}

Substituting \eqref{eq:approxu} into \eqref{eq:areaPolar2}, and exploiting the symmetric property of the lobe with respect to $x$-axis, we have
\begin{align}
&\mathcal{A}_m\approx 2 \gamma_p^{\frac{2}{3}}c_1^{-\frac{1}{3}}\int_{0}^{\pi/2}{\cos^{\frac{2}{3}}(\theta)}d\theta-\frac{2c_2}{3c_1}\int_{0}^{\pi/2}{\cos^2(\theta)}d\theta\nonumber\\
&\overset{(a)}{=}\gamma_p^{\frac{2}{3}}c_1^{-\frac{1}{3}}\sqrt{\pi}\frac{\Gamma(5/6)}{\Gamma(4/3)}-\frac{\pi}{6}\frac{c_2}{c_1},\label{eq:areainAB}
\end{align}
where $(a)$ follows from the integration formulas in \cite{InteFormula} and $\Gamma(\cdot)$ is the Gamma function. Substituting the expression of $c_1,c_2$ and evaluating the numerical value for constants in \eqref{eq:areainAB} gives \eqref{eq:MonoCoverage}.


\end{document}